\documentclass[reqno]{amsart}

\usepackage{caption}
\usepackage{amsaddr}
\usepackage{comment}
\usepackage[margin=1in]{geometry} 
\usepackage{graphicx,mathtools,tikz,hyperref}
\usepackage{amsthm, amssymb, amsmath}
\usetikzlibrary{quantikz}
\usepackage{fancyhdr}

\fancypagestyle{firstpage}{
  \fancyhf{} % clear header and footer
  \fancyfoot[C]{%
    \thepage\\[0.5ex]
    Distribution Statement A. Approved for public release: Distribution unlimited
  }

}

\usepackage{soul}
\usepackage[utf8]{inputenc}
\usepackage[english]{babel}
\usepackage{mathdots}
\usepackage{empheq}
\usepackage{enumitem}
\usepackage{theoremref}
\usepackage{pgfplots}
\pgfplotsset{compat = newest}
\usetikzlibrary{decorations.pathmorphing, patterns, pgfplots.fillbetween, intersections}
\usetikzlibrary{decorations.pathreplacing,angles,quotes}

\usetikzlibrary{positioning}

\numberwithin{equation}{section}
\theoremstyle{plain}
\newtheorem{theorem}{Theorem}[section]

\newtheorem{lemma}[theorem]{Lemma}
\theoremstyle{remark}
\newtheorem{remark}[theorem]{Remark}
\theoremstyle{definition}

\date{\today}

\title{A Linear Combination of Unitaries Decomposition for the Laplace Operator}
\author{Thomas Hogancamp}
\email{thomas.e.hogancamp.ctr@us.navy.mil}
\address{U.S. Naval Research Laboratory, Monterey, CA, 93943, United States and American Society for Engineering Education, Washington, D.C.}
\author{Reuben Demirdjian}

\address{U.S. Naval Research Laboratory, Monterey CA, 93943, United States}
\author{Daniel Gunlycke}
\address{U.S. Naval Research Laboratory, Washington, D.C., 20375, United States}

\begin{document}

\thispagestyle{firstpage}
\begin{abstract}
   We provide novel linear combination of unitaries decompositions for a class of discrete elliptic differential operators. Specifically, Poisson problems augmented with periodic, Dirichlet, Neumann, Robin, and mixed boundary conditions are considered on the unit interval and on higher-dimensional rectangular domains. The number of unitary terms required for our decomposition is independent of the number of grid points used in the discretization and scales linearly with the spatial dimension. Explicit circuit constructions for each unitary are given and their complexities analyzed. The worst case depth and elementary gate cost of any such circuit is shown to scale at most logarithmically with respect to number of grid points in the underlying discrete system. We also investigate the cost of using our method within the Variational Quantum Linear Solver algorithm and show favorable scaling. Finally, we extend the proposed decomposition technique to treat problems that include first-order derivative terms with variable coefficients. 
\end{abstract}
\maketitle
 
\section{Introduction}%%%%%%%%%%%%%%%%%%%%%%%%%%%%%%%%%%%%%%%%%%%%%%%%%%%%

The importance of Partial Differential Equations (PDEs) in the physical sciences cannot be overstated. They lay the mathematical foundations of diffusion, electromagnetism, fluid dynamics, elasticity, quantum mechanics, and relativity. Elliptic PDEs are a specialized class of such equations that are used to model steady states encountered in these subject areas. Moreover, the Laplace operator, which is the prototypical elliptic differential operator, also appears in many of their associated time-dependent formulations. Poisson's equation is a fundamental elliptic equation as it involves only the Laplace operator, a forcing term, and boundary conditions. The relative simplicity of Poisson problems combined with their central role in modeling real-world physical systems makes them an excellent test bed for novel PDE methods. The results obtained in this setting can often be adapted to treat more complex problems, including those with general elliptic operators or time dependence such as heat or wave equations. 

Explicit solutions for PDEs are exceedingly rare, and real-world applications typically rely on numerical methods to generate faithful approximations. The Finite Difference Method (FDM) is a foundational approach to solving discretized PDEs including Poisson's equation. The FDM involves discretization of the spatial domain, approximating derivatives with finite differences, encoding forcing terms and any boundary data, and finally setting up a system of equations for the corresponding discrete values of the unknown solution. This technique is rather straightforward in set-up and execution for Poisson problems posed on rectangular domains. However, the advantages of simplicity can be offset by poor performance when the FDM is implemented on classical hardware. For example, the most expensive subroutine for many Navier-Stokes solvers involves solving a FDM Poisson equation for a correction pressure term, and research aimed at improving this process is ongoing \cite{costa2018fft, li2022new}. 

Exciting developments in quantum computing have opened the possibility for quantum algorithms capable of solving large systems of equations with an exponential improvement over classical algorithms. The Harrow-Hassidim-Lloyd algorithm \cite{PhysRevLett.103.150502} and its variants \cite{ambainis2012variable, chakraborty2019power, childs2017quantum} are among the most popular quantum linear system solvers. However, these algorithms can be restrictive in two nontrivial ways; they rely on nondescript ``oracles" to perform several key tasks that complicate rigorous analysis, and they are ill-suited for the noisy intermediate-scale quantum (NISQ) devices currently available. Variational Quantum Algorithms (VQAs) offer an alternative approach to solving systems of equations and are considered good candidates for the NISQ era \cite{huang2021near}. The Variational Quantum Linear Solver (VQLS) \cite{bravo2023variational} is a hybrid algorithm that offloads expensive cost function evaluations to a quantum computer and uses the results to perform optimization on a classical machine. The parameters to be optimized describe trial solutions and the algorithm terminates when some pre-defined threshold of accuracy is achieved.  

The resource requirements for VQLS are critically dependent on the Linear Combination of Unitaries (LCU) decomposition used to encode the underlying linear system. Indeed, the number of circuit executions required per iteration of the algorithm depends quadratically on the number of terms in the LCU \cite{bravo2023variational}. Hence, a balance between the circuit depths and the total number of circuits in an LCU decomposition is sought for VQLS applications; naive decompositions using only very shallow depth circuits, such as tensor products of Pauli gates, typically require too many terms to be efficient. The number of unitary terms and their associated circuit depths can be at most $O(\textup{polylog}(N))$, where $N$ is the system size, to achieve efficiency in the VQLS setting. Since there is no general decomposition strategy that satisfies these criteria, researchers are faced with the serious challenge of finding viable candidates on a problem-specific case. 

The utility and importance of efficient LCU decompositions are not limited to their role within the VQLS framework. For instance, many quantum algorithms use LCU decompositions to block-encode matrices \cite{low2019hamiltonian,childs2012hamiltonian,chakraborty2019power, gilyen2019quantum, kharazi2024explicit, gnanasekaran2024efficient,gnanasekaran2025efficient}. A number of fault-tolerant linear system solvers rely on this process as a critical subroutine \cite{childs2017quantum, childs2021high}. Furthermore, recent developments in the linear combination of Hamiltonian simulation technique (LCHS) have produced promising near-optimal quantum differential equation solvers \cite{an2023linear, an2026quantum}. It seems likely that the demand for efficient LCU decompositions will only grow as the fault-tolerant era is approached. 

Quantum linear system solvers can potentially be combined with numerical methods, such as the FDM, to solve PDEs more efficiently, and with finer precision, than their classical counterparts. This exciting possibility has motivated a significant amount of recent research \cite{an2025quantum, an2023linear, an2026quantum, bae2024hardware,  berry2014high, liu2021variational,kharazi2024explicit, sato2021variational, childs2021high, demirdjian2022variational, demirdjian2025efficient, gnanasekaran2024efficient, gnanasekaran2025efficient}. However, the class of equations treated in the literature remains restrictive, and end-to-end applications are rare. In particular, the treatment of non-constant boundary and forcing functions, first-order derivatives, variable coefficients, complex geometries, and solution readout are largely unaddressed; the potential for real-world applications remains narrow until such features can be incorporated into quantum PDE solvers. Furthermore, quantum advantage is challenging to demonstrate even in the highly simplified models currently studied, and the development of efficient algorithms amenable to near-term devices is needed for testing and verification. 

In this work, we provide novel LCU decompositions for several classes of discrete elliptic operators obtained via the FDM. Poisson problems posed on rectangular domains in arbitrary dimension with periodic, Dirichlet, Neumann, Robin, and mixed boundary conditions are treated. The efficiency of our constructions are supported by an in-depth complexity analysis showing that they require at most $10d$ LCU terms with circuit depths and gate costs at most $O(\log{N})$, where the corresponding PDE is posed in $d$ dimensions and discretized using $N$ grid points. We believe these results should be useful for a variety of applications, since generating LCU decompositions for the Laplacian is a key subroutine of many quantum PDE solvers \cite{childs2021high,an2023linear, an2025quantum,an2026quantum, bae2024hardware,kharazi2024explicit, sato2021variational, demirdjian2022variational,demirdjian2025efficient}. Moreover, we consider the use of our constructions for a $1$-dimensional Poisson problem within the VQLS algorithm and compare the results with previously proposed LCU decompositions. Our method is shown to require at least $\log^{2}{N}$ fewer resources per iteration. Importantly, the decompositions presented in this article are not inherently tied to or limited by the details of the VQLS algorithm. For example, they could be adapted to LCU based fault-tolerant solvers in the spirit of \cite{childs2017quantum}. Given the prevalence of LCU based matrix decompositions in both NISQ era variational algorithms and fault-tolerant linear solvers, we believe our methods will be more widely useful in both near term and future quantum PDE solvers. 

Finally, we present an extension of our techniques to treat equations with variable-coefficient first-order derivative terms. The $1$-dimensional Dirichlet Poisson problem, augmented with a first-order derivative terms whose coefficient is a $k\textsuperscript{th}$ degree polynomial, is selected as a pilot problem. The resulting decomposition has $O(\log^{k}(N))$ unitary terms that can be implemented with circuit depth and gate cost at most $O(\log(N))$. To the best of our knowledge, this is the first explicit LCU decomposition strategy for Poisson problems with general polynomial coefficients. 

The remainder of this work is structured as follows. The basic notation used throughout this paper and an overview of VQLS are discussed in Section \ref{prelim_sect}. The mathematical formulation of the PDEs and their discretized counterparts to be considered are given in Section \ref{prob_sect}. Our main results are recorded as theorems in Section \ref{main_results_sect}. In Section \ref{Decomp_sect}, derivations of our novel LCU decompositions are presented. A complexity analysis for these decompositions is given in Section \ref{complexity_sect}. In addition, we investigate the resource requirements to implement our methods within the VQLS algorithm for a $1$-dimensional Poisson problem with Dirichlet boundary conditions. Section \ref{proof_sect} ties together the work of the preceding sections to give proofs of the main results. We provide a method to extend our LCU decompositions to treat problems with polynomial-coefficient first-order derivative terms in Section \ref{ext_sect}. Finally, we discuss our conclusions and future research directions in Section \ref{disc_sect}.  

\section{Preliminaries} \label{prelim_sect} %%%%%%%%%%%%%%%%%%%%%%%%%%%%%%%%%%%%%%%%%%%%%%%%%%%%%%%%%%%%%%
In this section, we discuss the notation to be used throughout the rest of the article. We also give a brief summary of the VQLS algorithm for the convenience of the reader. 
\subsection{Notation}%%%%%%%%%%%%%%%%%%%%%%%%%%%%%%%%%%%%%%%%%%%%%%%%%%%%
We denote by $U(N)$ the set of $N\times N$ unitary matrices over $\mathbb{C}$. Let us recall the matrix representations of some commonly used single-qubit gates in quantum computing: 
\begin{align*}
    I & \coloneqq \begin{pmatrix}
        1 & 0 \\
        0 & 1
    \end{pmatrix} \qquad\qquad \; \; 
    X  \coloneqq \begin{pmatrix}
        0 & 1 \\
        1 & 0
    \end{pmatrix} \\
    Z & \coloneqq \begin{pmatrix}
        1 & 0 \\
        0 & -1
    \end{pmatrix} \qquad \qquad
    H  \coloneqq \frac{1}{\sqrt{2}} \begin{pmatrix}
        1 & 1 \\
        1 & -1
    \end{pmatrix} \, . 
\end{align*}
Here, $X$ is the Pauli-$X$ gate (or $NOT$ gate), $Z$ is the Pauli-$Z$ gate, and $H$ is the Hadamard gate. We will also make use of the following matrices:
\begin{align} \label{tau_def}
\begin{split}
    \tau_{0} & \coloneqq\begin{pmatrix}
        1 & 0 \\
        0 & 0
    \end{pmatrix} 
    \qquad \tau_{1} \coloneqq\begin{pmatrix}
        1 & 0 \\
        0 & 0
    \end{pmatrix} \\
    \tau_{2} & \coloneqq \begin{pmatrix}
        0 & 0 \\
        1 & 0
    \end{pmatrix} 
    \qquad \tau_{3}  \coloneqq \begin{pmatrix}
        0 & 0 \\
        0 & 1
    \end{pmatrix} \, ,
\end{split}
\end{align}
and we refer to the set $\mathbb{T} = \{ \tau_{0}, \tau_{1}, \tau_{2}, \tau_{3} \}$ as the tau-basis. Note that the elements of $\mathbb{T}$ are not unitary matrices and therefore do not represent quantum gates. Nevertheless, they are useful when decomposing matrices and play an important role in this work. 

Next, we need to lay out notation for several controlled operations. Consider a quantum register with $k$ qubits indexed starting at $0$. Let $t \in \{1, \cdots , k-1\}$ and $c \in \{0,\dots, t-1\}$. Then, 
\begin{equation} \label{CNOT_def}
    CNOT(c,t) \coloneqq I^{\otimes c}\otimes \left(\tau_{0}\otimes I^{\otimes t-c}+\tau_{3}\otimes I^{\otimes t-c-1}\otimes X\right)\otimes I^{\otimes k-t-1} 
\end{equation}
defines a controlled $NOT$ operation with control $c$ and target $t$. The following basic multicontrol $Z$ and $NOT$ operations will also be used: 
\begin{equation} \label{control_defs_tz}
    C_{k}(Z) \coloneqq \begin{pmatrix}
            1&  &  &  &  \\
            & 1 &  &  &  \\
            & & \ddots & &  \\ 
            & & & 1& \\
            & & & & -1\\
            \end{pmatrix}_{2^{k}\times 2^{k}} 
        \qquad 
    T_{k} \coloneqq \begin{pmatrix}
            1&  &  &  &  \\
            & 1 &  &  &  \\
            & & \ddots & &  \\ 
            & & & 0&1 \\
            & & & 1& 0\\
            \end{pmatrix}_{2^{k}\times 2^{k}}       
        \, .
\end{equation}
Note that $C_{k}(Z)$ and $T_{k}(Z)$ each control the first $k-1$ qubits. We also need to establish notation for more complex multicontrol $NOT$ operations. We omit the relevant matrix representations as they are not required in the following analysis. Let $C \subset \{0,\cdots, k-1\}$, with $|C| \leq k-1$, and suppose that $t \in \{0, \cdots , k-1\} \setminus C$. Then, we use $CNOT(C, t)$ to denote the multicontrol $NOT$ operation with controls $C$ and target $t$.

Lastly, we define the increment gate whose matrix representation is 
\begin{equation} \label{inc_gate_def}
        S_{k} \coloneqq \begin{pmatrix}
            0&  & \hdots &  0& 1 \\
            1& 0 &  &  & 0 \\
            0& 1& \ddots & &\vdots  \\ 
            \vdots & \ddots& \ddots& 0& \\
            0&\hdots & 0&1 & 0\\
            \end{pmatrix}_{2^{k}\times 2^{k}} \, .
\end{equation}
The gate $S_{k}^{\dagger}$ is referred to as a decrement gate. 

\subsection{VQLS Overview}%%%%%%%%%%%%%%%%%%%%%%%%%%%%%%%%%%%%%%%%%%% 
\label{vqls_subsec}
The LCU decompositions to be presented in this work are general enough to be applicable in many scenarios and could be used in concert with a variety of linear solvers. We choose to investigate their use within the VQLS algorithm to demonstrate one such application and provide a framework for comparison with previously proposed decompositions. A rapid overview of the VQLS algorithm is given in this subsection for convenience. 

Let $A \in \mathbb{C}^{N\times N}$ and $\mathbf{b} \in \mathbb{C}^{N}$, where $N = 2^{n}$, and suppose that we want to solve 
\begin{equation*}
    A \mathbf{x} = \mathbf{b} \, . 
\end{equation*}
Suppose that we have access to a parametrized circuit $V(\theta)$ that produces trial states by $V(\theta)|0 \rangle = x(\theta)$ and also to a circuit $U_{b}$ such that $U_{b}|0\rangle = |b\rangle$, where $|b\rangle$ is proportional to $\mathbf{b}$. The goal of the VQLS algorithm is to optimize the parameter $\theta$ using an appropriate cost function to find some $\theta_{0}$ for which $A|x(\theta_{0})\rangle \propto |b\rangle$. Given an LCU decomposition of $A$ in the form 
\begin{equation*}
    A = \sum_{l=1}^{L}c_{l}R_{l} \, , 
\end{equation*}
where $c_{l} \in \mathbb{C}$ and $R_{l} \in U(2^{n})$, the \textit{local cost function} proposed in \cite{bravo2023variational} can be written as 
\begin{equation}\label{local_cost}
    C(\theta) = \frac{1}{2} - \frac{1}{2N}\frac{\sum_{j=1}^{n}\sum_{ll'}c_{l}c_{l'}^{*}\delta^{(j)}_{ll'}}{\sum_{ll'}c_{l}c_{l'}^{*}\beta_{ll'}} \, , 
\end{equation}
where
\begin{equation} \label{delta_def}
    \delta_{ll'}^{(j)} = \langle \mathbf{0} | V^{\dagger}(\theta) R_{l'}^{\dagger}U_{b} Z_{j} U_{b}^{\dagger} R_{l} V(\theta)|\mathbf{0} \rangle \, , 
\end{equation}
and 
\begin{equation} \label{beta_def}
    \beta_{ll'} = \langle \mathbf{0} | V^{\dagger}(\theta) R^{\dagger}_{l'} R_{l} V(\theta)| \mathbf{0} \rangle \, , 
\end{equation}
and $Z_{j} \coloneq I^{\otimes j-1}\otimes Z \otimes I^{\otimes n-j}$. Note that $\beta_{ll} =1$. Also, $\beta_{ll'} = \beta_{ll'}^{*} = \beta_{l'l}$ whenever all circuits involved are real-valued (this holds for $\delta^{(j)}_{ll'}$ as well). These observations allow us to skip repetitive computations and reduce the cost of evaluating $C(\theta)$. 

The Hadamard Test can be used for each iteration of the algorithm to estimate the required values of $\beta_{ll'}$ and $\delta^{(j)}_{ll'}$. Then, a classical computer is used to calculate the cost function and update the parameter $\theta$. If $c_{l}$, $R_{l}$, $\mathbf{b}$, and $V(\theta)$ are all real-valued, then the Hadamard Test circuits proposed in \cite{bravo2023variational}, and depicted below in Figures \ref{fig_1} and \ref{fig_2}, can be used. Note that these circuits must be run multiple times to obtain accurate estimates; see \cite{bravo2023variational} for details. 

\begin{figure}[h!] % [h!] suggests placing it 'here' if possible
    \centering
    \includegraphics[width=0.5\linewidth]{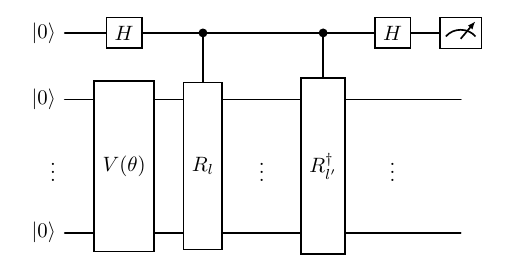}
    \caption{Hadamard Test Circuit for $\beta_{ll'}$ (real-valued case).}
    \label{fig_1}
\end{figure}
\begin{figure}[h!] % [h!] suggests placing it 'here' if possible
    \centering
    \includegraphics[width=0.5\linewidth]{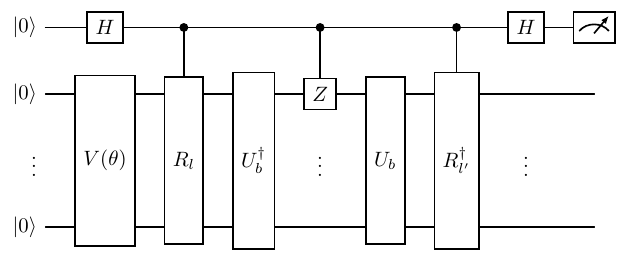}
    \caption{Hadamard Test Circuit for $\delta^{(1)}_{ll'}$ (real-valued case). The target of the controlled $Z$ gate must be modified appropriately for $j\neq 1$.}
    \label{fig_2}
\end{figure}
\section{The Problem} \label{prob_sect}%%%%%%%%%%%%%%%%%%%%%%%%%%%%%%%%%%%%%%%%%%%%%%%%%%%%%%%%%%%%%%%%%%%%%%

Let $\Omega \subset \mathbb{R}^{d}$ be an open bounded set. We are interested in the following Poisson problem: 
\begin{empheq}[left=\empheqlbrace]{align}
\label{basic_poiss}
\begin{split}
        \Delta u&=f \qquad \qquad \qquad \;\;\text{in} \; \Omega\\  
         \gamma \partial_{\nu}u +\beta u &=g \qquad \qquad \qquad \;\; \text{on} \; \partial \Omega \, , 
\end{split}
\end{empheq}
where $\gamma, \beta$ are constants, $\partial_{\nu}$ denotes the outward unit normal on $\Omega$, $f$ and $g$ are given functions, and $u$ is the desired solution. In addition to the general boundary condition given above, we also consider the periodic case. The literature regarding the existence and regularity of solutions for equation \eqref{basic_poiss} is vast. We mention only that these results are intimately connected to the regularity of $f$, $g$, and $\partial \Omega$. See \cite{gilbarg1977elliptic} and \cite{evans2022partial} for textbook treatments.   

The state-of-the-art applications for quantum computing are unable to capture the full generality of \eqref{basic_poiss}. Recent work in this area is restricted to solving either $1$-dimensional problems or those higher-dimensional ones for which $\Omega = (0,1)^{d}$. We adopt the same geometric constraints and leave the treatment of curved boundaries to future research.  

In the following subsections, we consider discretizations of the $1$-dimensional Poisson problem on $\Omega = (0,1)$ with a variety of boundary conditions. All cases are presented as $2^{n}\times 2^{n}$ linear systems. Lastly, we introduce some basic extensions to the higher-dimensional setting. 

\subsection{$1$-dimensional Dirichlet}%%%%%%%%%%%%%%%%%%%%%%%%%%%%%%%%%%
\label{1d_Dir_sub_sect}
First, let us consider the case of Dirichlet boundary conditions given by $u(0)=a$ and $u(1)=b$. Let $h_{1} = \frac{1}{2^{n}+1}$ and $x_{i} = ih_{1}$. In this case, a second-order central difference scheme can be used to approximate $\Delta$ inside $(0,1)$ and yields the following discretized Laplace operator:  
\begin{equation} \label{basic_LD_mat}
    L_{n,1} = \frac{1}{h_{1}^{2}}\begin{pmatrix}
            -2 & 1  & 0 & \cdots &  & 0 \\
            1  & -2 & 1 & 0 &  &  \vdots\\ 
            0  & 1 & -2 & 1 &  &  \\
        \vdots &   \ddots & \ddots  & \ddots   &  \ddots &  0 \\
              &       &  0  &  1 & -2 &  1 \\
            0  &       &  \cdots  &   0& 1 & -2 \\
            \end{pmatrix}_{2^{n}\times 2^{n}}\, , 
\end{equation}
and our problem takes the form 
\begin{equation}\label{disc_Dir_sys}
    L_{n,1} \bold{u} = \bold{f} +B_{1}, 
\end{equation}
where 
\begin{align*}
   \bold{u} &= (u(x_{1}), \dots, u(x_{2^{n}}))^{T} \\
   \bold{f} &= (f(x_{1}), \dots, f(x_{2^{n}}))^{T} \\ 
   B_{1} & = -\left(\frac{a}{h_{1}^{2}},0,\dots,0,\frac{b}{h^{2}_{1}}\right)^{T} \, . 
\end{align*}
Note that the system \eqref{disc_Dir_sys} sets up $2^{n}$ equations corresponding to the interior points of $(0,1)$ and the boundary values are accounted for through $B_{1}$. It can be shown that this scheme is second-order accurate. For convenience, we introduce a normalized matrix $A_{n,1}$ given by 
\begin{equation} \label{A_def}
    A_{n,1} = h_{1}^{2} L_{n,1} \, . 
\end{equation}
It should be emphasized that $A_{n,1}$ is invertible. 

\subsection{$1$-dimensional Periodic} %%%%%%%%%%%%%%%%%%%%%%%%%%%%%%%%%%%%%%%%%%%%%%%%%%%%

Periodic boundary conditions on the interval $(0,1)$ amount to the constraint that $u(0)=u(1)$. Let $h_{2} = \frac{1}{2^{n}}$ and $x_{i} = ih_{2}$. The periodic Laplace operator can be implemented with second order accuracy using 
\begin{equation} \label{basic_LP_mat}
    L_{n,2} = \frac{1}{h_{2}^{2}}\begin{pmatrix}
            -2 & 1  & 0 & \cdots & 0 & 1 \\
            1  & -2 & 1 & 0 & \cdots &  0\\ 
            0  & 1 & -2 & 1 &  & \vdots \\
        \vdots &   \ddots & \ddots  & \ddots   &  \ddots &  0 \\
            0  &       &   0 &  1 & -2 &  1 \\
            1  &   0    &  \cdots  &   0& 1 & -2 \\
            \end{pmatrix}_{2^{n}\times 2^{n}}\, . 
\end{equation}
The resulting discrete Poisson equation is 
\begin{equation*}
    L_{n,2}\bold{u} = \bold{f} \, ,
\end{equation*}
where $\bold{u} = (u(x_{0}), \dots, u(x_{2^{n}-1}))^{T}$ and $\bold{f} = (f(x_{0}), \dots, f(x_{2^{n}-1}))^{T}$. Note that this scheme includes the left boundary point while excluding the right one. We introduce the normalization 
\begin{equation*}
    A_{n, 2} = h_{2}^{2}L_{n,2}
\end{equation*}
and note that $A_{n,2}$ is a singular matrix. 

\subsection{$1$-dimensional Neumann/Robin}%%%%%%%%%%%%%%%%%%%%%%%%%%%%%%%%%%%%%%%%%%%%%%%%%%%%%%%%%%%%%%%%
\label{1d_NR_sub}
We treat the Neumann and Robin cases simultaneously for convenience. Robin boundary conditions on $(0,1)$ may be expressed as 
\begin{align} \label{Robin_cond}
    \begin{split}
        u_{x}(0)+a_{0}u(0) &= b_{0} \\ 
        u_{x}(1)+a_{1}u(1) & = b_{1} \, , 
    \end{split}
\end{align}
where $a_{0}$ and $a_{1}$ are non-zero constants. We can recover pure Neumann boundary conditions on one (or both) endpoints from \eqref{Robin_cond} by setting one (or both) of $a_{0}, a_{1}$ equal to $0$. Note that \eqref{Robin_cond} describes mixed Neumann--Robin boundary conditions if precisely one of $a_{0},a_{1}$ is equal to $0$.  

There are several ways to discretize the $1$-dimensional Poisson problem with Neumann/Robin boundary conditions as given by \eqref{Robin_cond}. We would like to retain second-order accuracy, as done with the periodic and Dirichlet cases, and accordingly adopt 
\begin{equation} \label{basic_LN_mat}
    L_{n,3} = \frac{1}{h_{3}^{2}}\begin{pmatrix}
            (-2+2a_{0}h_{3}) & 2  & 0 & \cdots &  & 0 \\
            1  & -2 & 1 & \ddots&  &  \vdots\\ 
            0  & 1 & -2 & 1 &  &  \\
        \vdots &    & \ddots  & \ddots   &  \ddots &  0 \\
              &       &   \ddots&  1 & -2 &  1 \\
            0  &       &  \cdots  &   0& 2 & (-2-2a_{1}h_{3}) \\
            \end{pmatrix}_{2^{n}\times 2^{n}}\, ,
\end{equation}
where $h_{3} = \frac{1}{2^{n}-1}$ and $x_{i} = i h_{3}$. The problem is reduced to 
\begin{equation*}
    L_{n,3}\bold{u} = \bold{f} + B_{3} \, , 
\end{equation*}
where 
\begin{align*}
   \bold{u} &= (u(x_{0}), \dots, u(x_{2^{n}-1}))^{T} \\
   \bold{f} &= (f(x_{0}), \dots, f(x_{2^{n}-1}))^{T} \\ 
   B_{3} & = \left(\frac{2b_{0}}{h_{3}},0,\dots,0,-\frac{2b_{1}}{h_{3}}\right)^{T} \, . 
\end{align*}
Note that the system \eqref{basic_LN_mat} includes equations for all interior points and both endpoints. Moreover, the equations corresponding the first and last rows of \eqref{basic_LN_mat} are derived using so-called ghost points.  

We will work with matrix
\begin{equation*}
    A_{n,3} = h_{3}^{2}L_{n,3}
\end{equation*}
when deriving an LCU for the Neumann/Robin case. The invertibility of $A_{n,3}$ is sensitive to the values of $a_{0}, a_{1}$, and $h_{3}$. The choice $a_{0}=a_{1}=0$, for example, results in a singular matrix. On the other hand, setting $a_{0} = -\dfrac{1}{h_{3}}$ and $a_{1}= - \dfrac{1}{h_{3}}$ yields a non-singular matrix. 

\subsection{Higher-dimensional problems}%%%%%%%%%%%%%%%%%%%%%%%%%%%%%%%%%%%%%%%%%%%%%
\label{Higher_D_Sub}
Let $d\geq 2$ be an integer. For a general domain $\Omega$, finite difference schemes for \eqref{basic_poiss} are far more complicated, if not intractable, when compared to the $d=1$ setting. However, there exist relatively straightforward extensions when $\Omega$ is taken to be a box-type domain. To this end, let $\Omega = \Pi_{i=1}^{d}(0,1)$, $\bold{n} = (n_{1}, \dots, n_{d})$, and $\bold{\alpha} = (\alpha_{1}, \dots, \alpha_{d})$, where $n_{i}>2$ is an integer and $\alpha_{i} \in \{1,2,3\}$. Our aim is to apply the operator $L_{n_{i},\alpha_{i}}$ to the $i$-th dimension of $\Omega$. For each $i \in \{1, \dots, d\}$, let $(x_{1}^{i}, \dots, x^{i}_{2^{n_{i}}})$ denote the discretization of $(0,1)$ corresponding to $L_{n_{i},\alpha_{i}}$ as described in the preceding subsections (this requires a trivial relabeling of the $x^{i}_{j}$ values whenever $\alpha_{i}\neq 1$). Now, we let
\begin{equation*}
    \bold{u} \coloneqq (u_{1,1,\cdots,1}, u_{2,1,\cdots,1},\dots, u_{N_{1},1,\cdots,1}, u_{1,2,1,\cdots,1},\dots, u_{1,1,\cdots, 1,N_{d}}, u_{2,1,\cdots,1,N_{d}}, \dots, u_{N_{1},\cdots ,N_{d}})^{T},
\end{equation*}
where $N_{i}=2^{n_{i}}$ and 
\begin{equation*}
    u_{i_{1},\cdots ,i_{d}} \coloneq u(x^{1}_{i_{1}}, \dots, x^{d}_{i_{d}})\, ,  
\end{equation*}
which can be seen as the vectorization of the multidimensional array $U$ whose entries are given by $U_{i_{1},\cdots, i_{d}}=u_{i_{1},\cdots, i_{d}}$. Finally, we let 
\begin{equation*}\label{higher_dim}
    L^{d}_{\bold{n},\bold{\alpha}} \coloneqq \sum_{i=1}^{d}I^{\otimes n_{d}}\otimes I^{\otimes n_{d-1}}\otimes \cdots\otimes  I^{\otimes n_{i+1}} \otimes L_{n_{i},\alpha_{i}} \otimes I^{\otimes n_{i-1}}\otimes\cdots \otimes I^{\otimes n_{1}} \, . 
\end{equation*}
The operator $L^{d}_{\bold{n},\bold{\alpha}} \in \mathbb{R}^{N\times N}$ acts on $\bold{u} \in \mathbb{R}^{N\times 1}$, where $N = N_{1}N_{2}\cdots N_{d}$, to approximate $\Delta u$ inside $\Omega$ while applying the boundary conditions associated to $L_{n_{i},\alpha_{i}}$ on $\partial \Omega \cap \{\bold{x} \in \mathbb{R}^{d} \; \; \lvert \; \; x^{i} =0,1\}$. Observe that mixed boundary conditions are obtained whenever the values of $\alpha_{i}$ are not all equal. 

Note that $L^{d}_{\bold{n},\bold{\alpha}}$ inherits the second-order accuracy of the constituent $1$-dimensional operators. We consider discretized Poisson problems of the form 
\begin{equation*}
    L^{d}_{\bold{n},\bold{\alpha}}\bold{u} = \bold{f}+B^{d}_{\bold{n},\alpha} \, , 
\end{equation*}
where $\bold{f}\in  \mathbb{R}^{N\times 1}$ is the vectorization of $f$ and $B^{d}_{n, \alpha}\in \mathbb{R}^{N\times 1}$ encodes the appropriate Dirichlet, Neumann, or Robin values analogously to the $d=1$ case. We do not provide an explicit representation of $B^{d}_{\bold{n},\alpha}$ in the general case since our primary focus involves LCU decompositions of differential operators. However, an explicit two-dimensional is provided in Appendix \ref{Append_B} to illustrate the basic idea. The invertiblity of $L^{d}_{\bold{n},\bold{\alpha}}$ is sensitive to each constituent operator and we make no claims about the general situation. 

\section{Main Results} \label{main_results_sect} %%%%%%%%%%%%%%%%%%%%%%%%%%%%%%%%%%%%%%%%%%%%%%%%%%%%%%%%%%%%%%%%%

The primary contributions of this work are novel LCU decompositions for a wide class of Poisson problems. We record several key results in this section detailing their structure and complexity. Our first theorem concerns $1$-dimensional Dirichlet problems as outlined in Section \ref{1d_Dir_sub_sect} (see \eqref{basic_LD_mat} and \eqref{A_def}). 

\begin{theorem}{(Dirichlet Case)} \thlabel{thm1}
    The matrix $A_{n,1}$ admits a decomposition of the form
    \begin{equation} \label{thm1_eq}
        A_{n,1} = \sum_{l=1}^{5}c_{1,l}R_{1,l}\, , 
    \end{equation}
    where $c_{1,l} \in \{1, -2,\pm \frac{1}{2} \}$ and each $R_{1,l}$ is unitary. Moreover, if a single ancilla qubit is available, then a given $R_{1,l}$ may be implemented using at most $2$-Hadamard gates and $O(n)$ Toffoli, $X$, and $CNOT$ gates.
\end{theorem} 
\begin{remark}
    Two of the $R_{1,l}$ terms above are simply tensor products of Pauli matrices (Pauli terms).  
\end{remark}
See equation \eqref{DL3} for the explicit LCU decomposition of $A_{n,1}$. We do not provide a theorem for the periodic case, since it was thoroughly covered in previous works \cite{sato2021variational, childs2021high,kharazi2024explicit}. The representation of $A_{n,2}$ adopted in this paper is given in \eqref{periodic_LCU}. Our results for the $1$-dimensional Neumann/Robin case formulated in Section \ref{1d_NR_sub} are given below. 
\begin{theorem}{(Neumann/Robin Case)} \thlabel{thm2}
    The matrix $A_{n,3}$ admits a decomposition of the form
    \begin{equation} \label{thm2_eq}
        A_{n,3} = \sum_{l=1}^{10}c_{3,l}R_{3,l}\, , 
    \end{equation}
    where $c_{3,l} \in \mathbb{R} $ and each $R_{3,l}$ is unitary. Moreover, if a single ancilla qubit is available, then a given $R_{3,l}$ may be implemented using at most a single $Z$ gate, $2$-Hadamard gates, and $O(n)$ Toffoli, $X$, and $CNOT$ gates. 
\end{theorem} 
\begin{remark}
    We have suppressed the dependence of the coefficients $c_{3,l}$ on $n$ and the parameters of the relevant boundary operators (see \eqref{Robin_cond}). We note that some values of $c_{3,l}$ may vanish. For example, only $7$ nontrivial terms are required for pure Neumann boundary conditions.  
\end{remark}
The exact decomposition of $A_{n,3}$ is given in \eqref{NR_LCU}. The following theorem is an extension of \thref{thm1} and \thref{thm2} to the higher-dimensional setting. It is valid for box-type domains and the boundary conditions described in Section \ref{Higher_D_Sub}. 

\begin{theorem}{(Higher-Dimensions)} \thlabel{thm3}
    Let $d\geq 2$ be an integer, $\bold{n} = (n_{1}, \dots, n_{d})$, $n^{*} = \max\{n_{1}, \dots, n_{d}\}$, and $\bold{\alpha} = (\alpha_{1}, \dots, \alpha_{d})$, where each $n_{i}>2$ is an integer and $\alpha_{i} \in \{1,2,3\}$. Then, the matrix $L^{d}_{\bold{n},\bold{\alpha}}$ admits a decomposition of the form
    \begin{equation} \label{thm3_eq}
       L^{d}_{\bold{n}, \bold{\alpha}} = \sum_{i=1}^{d}\sum_{j=1}^{L_{\alpha_{i}}}c^{i}_{\alpha_{i},j}R^{i}_{\alpha_{i},j}\, , 
    \end{equation}
    where $L_{\alpha_{i}} \leq 10$, $c^{i}_{\alpha_{i},j} \in \mathbb{R}$, and each $R^{i}_{\alpha_{i},j}$ is unitary. Moreover, a given $R^{i}_{\alpha_{i},j}$ may be implemented using at most a single $Z$ gate, $2$-Hadamard gates, and $O(n^{*})$ Toffoli, $X$, and $CNOT$ gates.
\end{theorem} 
\begin{remark}
    Unlike \thref{thm1} and \thref{thm2}, this result requires no ancilla. 
\end{remark}
Both $c^{i}_{\alpha_{i},j}$ and $R^{i}_{\alpha_{i},j}$ are related to the $1$-dimensional cases in straightforward way through \eqref{High_D_rel_low}. Equations  \eqref{higher_dim_LCU} and \eqref{High_D_rel_low} describe an explicit representation of $L^{d}_{\bold{n}, \bold{\alpha}}$.

%%%%%%%%%%%%%%%%%%%%%%%%%%%%%%%%%%%%%%%%%%%
\section{LCU Decompositions} \label{Decomp_sect}

In this section, we derive unitary expansions for each $A_{n,j}$. Crucially, the individual terms of our decompositions can be implemented efficiently, as shown in Section \ref{complexity_sect}. The proposed decompositions appear to be new for all cases except $j=2$. We conclude the section by extending these results to higher-dimensional operators of the form $L^{d}_{\bold{n}, \bold{\alpha}}$.  

\subsection{Periodic case}
Several authors have proposed $3$ term LCU decompositions for $A_{n,2}$ \cite{sato2021variational, childs2021high, kharazi2024explicit}. In \cite{childs2021high} and \cite{kharazi2024explicit}, the following decomposition for the $1$-dimensional Laplacian with periodic boundary conditions is considered: 
\begin{equation}\label{periodic_LCU}
    A_{n,2} = -2I + S_{n}+S_{n}^{\dagger} \, , 
\end{equation}
where $S_{n}$ is the increment gate defined in \eqref{inc_gate_def}. The authors of \cite{kharazi2024explicit} highlight that the matrices $S_{n}$ and $S_{n}^{\dagger}$ can be implemented using $O(n)$ Toffoli, $CNOT$, and $X$ gates by the construction developed in \cite{gidney2015constructing}. Note that this method requires a single ancilla qubit. This approach marks a significant advantage over the more straightforward ancilla-free implementation described in \cite{douglas2009efficient}, which requires $O(n^{2})$ such gates. 

\subsection{Dirichlet case}%%%%%%%%%%%%%%%%%%%%%%%%%%%%
\label{dirichlet_decomp_fund}
We now introduce a novel decomposition for the $1$-dimensional Dirichlet Laplacian. Consider the following matrices 
\begin{equation}\label{Cpm}
    C_{n} \coloneqq \begin{pmatrix}
            &  &  &  & 1 \\
            &  &  & 0 &  \\
            & & \iddots & &  \\ 
            & 0& & & \\
            1& & & & \\
            \end{pmatrix}_{2^{n}\times 2^{n}}\, , \qquad  \qquad 
              C_{n}^{-} \coloneq \begin{pmatrix}
            &  &  &  & 1 \\
            &  &  & -1 &  \\
            & & \iddots & &  \\ 
            & -1& & & \\
            1& & & & \\
            \end{pmatrix}_{2^{n}\times 2^{n}}
\end{equation}
and note that 
\begin{equation} \label{Cnd}
   C_{n} = \frac{1}{2}C_{n}^{-}+\frac{1}{2}X^{\otimes n} \, . 
\end{equation}
It is clear from \eqref{basic_LD_mat}, \eqref{basic_LP_mat}, and \eqref{Cnd} that 
\begin{equation}\label{A2d}
    A_{n,1} = A_{n,2} - C_{n} = A_{n,2}-\frac{1}{2}X^{\otimes n}-\frac{1}{2}C^{-}_{n} \, .  
\end{equation}
$A_{n,2}$ can be handled using \eqref{periodic_LCU}, but we are left to represent $C_{n}^{-}$ using standard quantum gates. We claim that  
\begin{equation} \label{basic_decomp_1}
    -C^{-}_{n} = \left( X \otimes I^{\otimes n-1}\right)\cdot B \cdot  \left(X  \otimes I^{\otimes n-1}\right)\cdot  \left(I\otimes C_{n-1}(Z)\right)\cdot B \cdot \left(X\otimes I^{\otimes n-1}\right) \, ,
\end{equation}
where 
\begin{equation*}
    B  \coloneqq \prod_{t=1}^{n-1}CNOT(0,t).
\end{equation*}
To verify the claim, observe first that $B$ can be represented in matrix form by
\begin{equation}\label{B_k_exp}
    B =  \begin{pmatrix}
        I^{\otimes n-1 } &  \\ 
                           &  X^{\otimes n-1}
    \end{pmatrix} \, , 
\end{equation}
which can be checked by direct multiplication of the explicit $CNOT(0,t)$ operators given by \eqref{CNOT_def}. Also, the relations
\begin{equation}\label{int1}
    \left(X \otimes I^{\otimes n-1}\right)\cdot B\cdot  \left(X\otimes I^{\otimes n-1}\right) = \begin{pmatrix}
        X^{\otimes n-1} & \\
                        & I^{\otimes n-1}
    \end{pmatrix}
\end{equation}
and 
\begin{equation}\label{int2}
     B\cdot \left(X\otimes I^{\otimes n-1}\right)
    =  \begin{pmatrix}
        & I^{\otimes n-1} \\
        X^{\otimes n-1} &
    \end{pmatrix}
\end{equation}
follow from some straightforward matrix algebra and substitution of $B$ using \eqref{B_k_exp}. Starting from the right hand side of \eqref{basic_decomp_1}, we combine \eqref{int1} and \eqref{int2} to yield
\begin{align*}
 \begin{pmatrix}
        X^{\otimes n-1} & \\
                        & I^{\otimes n-1}
    \end{pmatrix} \cdot \left(I\otimes C_{n-1}(Z)\right)\cdot \begin{pmatrix}
        & I^{\otimes n-1} \\
        X^{\otimes n-1} &
    \end{pmatrix} &= \begin{pmatrix}
        X^{\otimes n-1} & \\
                        & I^{\otimes n-1}
    \end{pmatrix} \cdot \begin{pmatrix}
        C_{n-1}(Z) & \\
                        & C_{n-1}(Z)
    \end{pmatrix} \cdot \begin{pmatrix}
        & I^{\otimes n-1} \\
        X^{\otimes n-1} &
    \end{pmatrix} \\
     &= \begin{pmatrix}
          &   &        &   &-1 \\
            &  &        &1   & \\
            &   & \iddots &   & \\
            & 1  &        &  & \\
           -1 &   &        &   &  
     \end{pmatrix}_{2^{n}\times 2^{n}} \\
     &= - C^{-}_{n} \, , 
\end{align*}
as desired, where the third line is a consequence of the definition in \eqref{Cpm}.

Finally, we combine \eqref{periodic_LCU}, \eqref{A2d}, and \eqref{basic_decomp_1} to conclude that 
\begin{align} \label{DL3}
\begin{split}
A_{n,1} &= -2I^{\otimes n}+S_{n}+S_{n}^{\dagger}-\frac{1}{2}X^{\otimes n} \\ &+\frac{1}{2}\left(X \otimes I^{\otimes n-1}\right)\cdot B \cdot  \left(X  \otimes I^{\otimes n-1}\right)\cdot  (I\otimes C_{n-1}(Z))\cdot B \cdot \left(X\otimes I^{\otimes n-1}\right) \, .
\end{split}
\end{align}
Equation \eqref{DL3} is our fundamental decomposition of the $1$-dimensional Dirichlet Laplacian and will be shown to satisfy all the criteria of \thref{thm1}.

\subsection{Neumann/Robin case}%%%%%%%%%%%%%%%%%%%%%%%%%%%%%%%

The techniques developed in the previous subsection for the Dirichlet Laplacian are now used to derive an efficient decomposition for the Neumann/Robin case. It follows from \eqref{basic_LD_mat} and \eqref{basic_LN_mat} that 
\begin{equation}\label{nm1}
    A_{n, 3} = A_{n,1} + \begin{pmatrix}
        2a_{0}h_{3} & 1 & 0 & \cdots & 0\\
        0   &  0 &   &   &  \\
        \vdots    &  \ddots &  \ddots &   & \vdots   \\
            &   &   & 0  &  0  \\
         0  &  \cdots &  0 & 1 & -2a_{1}h_{3} 
    \end{pmatrix}_{2^{n}\times 2^{n}} \, .
\end{equation}
Let $E_{n} = A_{n,3} - A_{n,1}$. Equations \eqref{tau_def}, \eqref{Cpm}, \eqref{nm1} and some quick calculations will verify that
\begin{equation}\label{E_n+1}
    E_{n} = 2a_{0}h_{3} \tau_{0}^{\otimes n }-2a_{1}h_{3}\tau_{3}^{\otimes n} + X^{\otimes n} \cdot C_{n} \cdot \left(I^{\otimes n-1}\otimes X\right) \, .
\end{equation}
An efficient decomposition for tensor products of matrices belonging to $\mathbb{T}$ is presented in Appendix \ref{Apen_A}. In particular, equations \eqref{tau_zero} and \eqref{tau_3} show that 
\begin{align*}
    \tau_{0}^{\otimes n} &= \frac{1}{4}I^{\otimes n}+\frac{1}{4}X^{\otimes n}\cdot C_{n}^{-}+\frac{1}{4}X^{\otimes n}\cdot C^{-}_{n}\cdot I^{\otimes n-1}\otimes Z+\frac{1}{4}I^{\otimes n-1}\otimes Z \\
    \tau_{3}^{\otimes n} &= \frac{1}{4}I^{\otimes n}+\frac{1}{4}X^{\otimes n}\cdot C_{n}^{-}-\frac{1}{4}X^{\otimes n}\cdot C^{-}_{n}\cdot I^{\otimes n-1}\otimes Z-\frac{1}{4}I^{\otimes -1}\otimes Z \, , 
\end{align*}
which can be combined with \eqref{E_n+1} to yield 
\begin{align} \label{En_explicit}
\begin{split}
    E_{n} &= \frac{(a_{0}h_{3}-a_{1}h_{3})}{2}I^{\otimes n}+\frac{(a_{0}h_{3}-a_{1}h_{3})}{2}X^{\otimes n}\cdot C_{n}^{-} +\frac{(a_{0}h_{3}
    +a_{1}h_{3})}{2}X^{\otimes n}\cdot C^{-}_{n}\cdot \left(I^{\otimes n-1}\otimes Z\right) \\ &+\frac{(a_{0}h_{3}+a_{1}h_{3})}{2}I^{\otimes n-1}\otimes Z + \frac{1}{2}I^{\otimes n-1}\otimes X+ \frac{1}{2}X^{\otimes n} \cdot C^{-}_{n} \cdot \left(I^{\otimes n-1}\otimes X\right) \, .
\end{split}
\end{align}
Finally, equations \eqref{nm1} and \eqref{En_explicit} can be combined to derive our LCU for the Neumann/Robin case: 
\begin{align} \label{NR_LCU}
    \begin{split}
        A_{n,3} &= \frac{(-4+(a_{0}-a_{1})h_{3})}{2}I^{\otimes n} + \frac{1}{2}I^{\otimes n-1}\otimes X - \frac{1}{2}X^{\otimes n}+\frac{(a_{0}+a_{1})h_{3}}{2}I^{\otimes n-1}\otimes Z \\ 
        &+ S_{n} + S_{n}^{\dagger}-\frac{1}{2}C_{n}^{-}+\frac{(a_{0}-a_{1})h_{3}}{2}X^{\otimes n}\cdot C_{n}^{-}\\ &+\frac{(a_{0}+a_{1})h_{3}}{2}X^{\otimes n}\cdot C_{n}^{-}\cdot \left(I^{\otimes n-1}\otimes Z\right) 
        + \frac{1}{2}X^{\otimes n}\cdot C_{n}^{-}\cdot \left(I^{\otimes n-1}\otimes X \right)\, . 
    \end{split}
\end{align}
Equation \eqref{NR_LCU} involves the sum of at most $10$ unitary terms (the exact number of terms depends on the values of $a_{0},a_{1}$, and $h_{3}$). Recall that an explicit representation of $C_{n}^{-}$ using standard quantum gates is given by equation \eqref{basic_decomp_1}. We leave $C_{n}^{-}$ in equation \eqref{NR_LCU} with the understanding that it will be substituted with \eqref{basic_decomp_1} in the proceeding analysis. 

\subsection{Higher-Dimensional Problems}%%%%%%%%%%%%%%%%%%%%%%%%%%%%%%%%%%%%%

A representation of the form \eqref{thm3_eq} is now easily obtained from our work in the $1$-dimensional setting. Indeed, using the set-up proposed in Section \ref{Higher_D_Sub}, we find that 
\begin{align} \label{higher_dim_LCU}
    \begin{split}
           L^{d}_{\bold{n},\bold{\alpha}} &= \sum_{i=1}^{d}I^{\otimes n_{d}}\otimes I^{\otimes n_{d-1}}\otimes \cdots \otimes I^{\otimes n_{i+1}}\otimes L_{n_{i},\alpha_{i}} \otimes I^{\otimes n_{i-1}}\otimes \cdots \otimes I^{\otimes n_{1}} \\
           &= \sum_{i=1}^{d}\sum_{j=1}^{L_{\alpha_{i}}}\frac{c_{\alpha_{i},j}}{h^{2}_{\alpha_{i}}}I^{\otimes n_{d}}\otimes I^{\otimes n_{d-1}}\otimes \cdots \otimes I^{\otimes n_{i+1}}\otimes R_{\alpha_{i},j} \otimes I^{\otimes n_{i-1}}\otimes\cdots \otimes I^{\otimes n_{1}} \, , 
    \end{split}
\end{align}
where $L_{1} = 5$, $L_{2} =3$, and $L_{3} = 10$. The second line of \eqref{higher_dim_LCU} follows from substituting each $ L_{n_{i},\alpha_{i}}$ using an appropriate choice of \eqref{periodic_LCU}, \eqref{DL3}, or \eqref{NR_LCU} (note that we convert from $A_{n_{i},\alpha_{i}}$ to $L_{n_{i},\alpha_{i}}$ via a multiplicative factor of $h_{\alpha_{i}}^{-2}$). For consistency with \eqref{thm3_eq}, we define the following: 
\begin{align} \label{High_D_rel_low}
\begin{split}
    c^{i}_{\alpha_{i},j} &\coloneq \frac{c_{\alpha_{i},j}}{h_{\alpha_{i}}^{2}} \\ 
    R^{i}_{\alpha_{i},j} & \coloneq I^{\otimes n_{d}}\otimes I^{\otimes n_{d-1}}\otimes \cdots \otimes I^{\otimes n_{i+1}}\otimes R_{\alpha_{i},j} \otimes I^{\otimes n_{i-1}}\otimes\cdots \otimes I^{\otimes n_{1}} \, .
\end{split}
\end{align}

\section{Complexity Analysis}%%%%%%%%%%%%%%%%%%%%%%%%%%%%%%%%%%%%%%%%%%%%%%%%%%%%%%%%%%%%%%
\label{complexity_sect}
Explicit LCU decompositions for several classes of Poisson problems were derived in Section \ref{Decomp_sect}. Upper bounds on the gate count and the depth of the circuits corresponding to individual unitary terms are described in this section. To accomplish this, we make the standing assumption that a single ancilla qubit is available in the $1$-dimensional setting and utilize well-known constructions for both increment gates and multicontrol $NOT$ gates. We conclude this section with a comparison, in terms of gate cost, of our Dirichlet LCU decomposition and that of \cite{kharazi2024explicit} in the context of the VQLS algorithm. 

\subsection{Resource Estimates for LCU Terms}%%%%%%%%%%%%%%%%%%%%%%%%%%%%%%%%%%%%%%%%%%%%%%%%%

We begin our analysis with the $1$-dimensional periodic, Dirichlet, and Neumann/Robin cases. We note that an analogous result for the periodic decomposition utilized in this paper was already given in \cite{kharazi2024explicit}. 

\begin{lemma}{($1$-dimensional LCU Cost)}\thlabel{lm1DLCU}
Given access to a single ancilla qubit, an element of the LCU decomposition for $A_{n,1}$, $A_{n,2}$, or $A_{n,3}$,  given by \eqref{DL3}, \eqref{periodic_LCU}, or \eqref{NR_LCU}, respectively, can be implemented using at most a single $Z$ gate, $2$-Hadamard gates, and $O(n)$ $X$, $CNOT$, and Toffoli gates. Moreover, this implementation can be accomplished with a circuit whose depth is at most $O(n)$.   
\end{lemma}

\begin{proof}
    First, let us consider $A_{n,1}$. Recall from \eqref{DL3} that the nontrivial unitary terms for the LCU decomposition of $A_{n,1}$ are $S_{n}$, $S_{n}^{\dagger}$, and $C_{n}^{-}$ (the others being $I^{\otimes n}$ and $X^{\otimes n}$). As shown in \cite{gidney2015constructing}, the increment gate $S_{n}$ can be implemented using $O(n)$ Toffoli, $X$, and $CNOT$ gates (it is here that access to the ancilla qubit is used). Moreover, the depth of the corresponding circuit is $O(n)$. The cost of $S^{\dagger}_{n}$ is the same. 

    We are left to consider $C_{n}^{-}$. Recall that $-C^{-}_{n} = \left(X \otimes I^{\otimes n-1}\right)\cdot B \cdot  \left(X  \otimes I^{\otimes n-1}\right)\cdot  \left(I\otimes C_{n-1}(Z)\right)\cdot B \cdot \left(X\otimes I^{\otimes n-1}\right)$ by equation \eqref{basic_decomp_1}. It is clear that this term requires $3$-$X$ gates, $2n$-$CNOT$ gates, and a single $C_{n-1}(Z)$ gate. Note that 
    \begin{equation}
    I\otimes C_{n-1}(Z) = \left(I^{\otimes n-1}\otimes H\right) \cdot \left(I\otimes T_{n-1} \right)\cdot \left(I^{\otimes n-1}\otimes H \right)\, , 
    \end{equation}
    where $T_{n-1}$ is defined in \eqref{control_defs_tz}. As shown in the seminal work \cite{barenco1995elementary}, a $CNOT(C,n-1)$ gate admits a decomposition into $O(n)$-Toffoli gates, whose associated circuit has depth $O(n)$, provided $|C|\leq n-2$ (so that at least a single borrowable bit is available).  See Lemmas 7.2 and 7.3 in \cite{barenco1995elementary} for more details. Hence, we find that $C_{n}^{-}$ can be constructed using $O(n)$-Toffoli gates with a circuit of linear depth, since $I \otimes T_{n-1}$ is equivalent to $CNOT(\{1,2,\dots, n-2\}, n-1)$. This completes the argument for $A_{n,1}$.

    Next, we consider $A_{n,2}$. Observe that each unitary term on the right hand side of the $A_{n,2}$ decomposition given in \eqref{periodic_LCU} also belongs to the LCU decomposition of $A_{n,1}$. Hence, the estimate for $A_{n,1}$ is obviously applicable to $A_{n,2}$ as well. Finally, we consider $A_{n,3}$. It is clear from \eqref{basic_decomp_1}, \eqref{DL3}, and \eqref{NR_LCU} that any element of the Neumann/Robin expansion not contained in the LCU for $A_{n,1}$ can be implemented using at most a single $Z$ gate, $O(n)$-$X$ gates, and a $C_{n}^{-}$ gate. The desired result now follows by invoking the cost estimate of $C_{n}^{-}$ derived above. 
\end{proof}

Several remarks concerning \thref{lm1DLCU} are in order. First, this result should be seen only as a basic upper bound on the resource requirements for our LCU decomposition. We choose to implement the relevant circuits using a single ancilla qubit and popular decompositions for both multicontrol $NOT$ gates \cite{barenco1995elementary} and increment gates \cite{gidney2015constructing}. Second, realization of our most expensive gates can be optimized in various ways for space or depth. For example, by utilizing an additional $O(n)$ ancilla qubits one can reduce the depths of implementing our constructions to $O(\log{n})$; Proposition 2 of \cite{moore2001parallel} could be used to treat the gate $B$, Corollary 2 of \cite{he2017decompositions} for $C_{n-1}(Z)$, and either Proposition 1 of \cite{moore2001parallel} (to achieve constant depth) or a modified adder from \cite{draper2004logarithmic} (to achieve $O(\log{n})$ depth) could be applied to $S_{n}$. 

Our arguments readily extend to the higher-dimensional decompositions presented in Section \ref{Decomp_sect}. We emphasize that no ancilla qubit is required for the following estimates. Also, the remarks of the previous paragraph apply to this higher-dimensional case and are worth keeping in mind. In particular, there may be many locally borrowable qubits in the higher-dimensional setting, and significant reductions in depth are possible for appropriate parameter regimes.

\begin{lemma}{(Higher-Dimensional LCU Cost)} \thlabel{HD_lem}
An element of the LCU decomposition for $L^{d}_{\bold{n},\bold{\alpha}}$  given by \eqref{higher_dim_LCU} can be implemented using at most a single $Z$ gate, $2$-Hadamard gates, and $O(n^{*})$ $X$, $CNOT$, and Toffoli gates, where $n^{*} = \max\{n_{1},\dots, n_{d}\}$. Moreover, this implementation can be accomplished with a circuit whose depth is at most $O(n^{*})$.  
\end{lemma}
\begin{proof}
    Recall that \eqref{higher_dim_LCU} gives us the following LCU decomposition of $L^{d}_{\bold{n},\alpha}$:
    \begin{equation*}
       L^{d}_{\bold{n},\bold{\alpha}} = \sum_{i=1}^{d}\sum_{j=1}^{L_{\alpha_{i}}}\frac{c_{\alpha_{i},j}}{h^{2}_{\alpha_{i}}}I^{\otimes n_{d}}\otimes I^{\otimes n_{d-1}}\otimes \cdots \otimes I^{\otimes n_{i+1}}\otimes R_{\alpha_{i},j} \otimes I^{\otimes n_{i-1}}\otimes\cdots \otimes I^{\otimes n_{1}} \,. 
    \end{equation*}
    Hence, we are left to consider a general term of the form $I^{\otimes n_{d}}\otimes I^{\otimes n_{d-1}}\otimes \cdots \otimes I^{\otimes n_{i+1}}\otimes R_{\alpha_{i},j} \otimes I^{\otimes n_{i-1}}\otimes\cdots \otimes I^{\otimes n_{1}}$. Note that each such term is guaranteed to have a borrowable ancilla qubit available due to its tensor product structure (in fact, it will have $n_{1}+\cdots n_{i-1}+n_{i+1}+\cdots+n_{d})$). Therefore, the the proof of \thref{lm1DLCU} implies that the gate cost and worst case circuit depth of an $R^{i}_{\alpha_{i},j}$ term is at most that of $R_{\alpha_{i},j}$. Moreover, it was established in the proof of \thref{lm1DLCU} that these quantities are at most $O(n_{i})$ and the result follows. 
\end{proof}

\begin{table}[h]
    \centering
    \caption{LCU Resource Estimates Summary}
    \label{tab:sample_table}
    \begin{tabular}{|l|c|c|c|c|c|}
        \hline
        & \textbf{\# of Qubits} & \textbf{\# of Terms in LCU} & \textbf{\# of Ancilla} & \textbf{Gate Cost} & \textbf{Depth} \\
        \hline
        \textbf{$1$-dim. Dirichlet} & $n$ & $5$ & $1$& $O(n)$ & $O(n)$  \\
        \hline
        \textbf{$1$-dim. Neumann/Robin} & $n$ & $10$ & $1$ & $O(n)$ & $O(n)$  \\
        \hline 
        \textbf{$d$-dim. Mixed} &$\sum_{i=1}^{d}n_{i}$ & $10d$ & $0$ & $O(n^{*})$ & $O(n^{*})$  \\[1pt]
        \hline
    \end{tabular}
   \\[10pt]
   \caption*{The worst case gate cost (in terms of Toffoli and Clifford gates) and circuit depth of individual LCU terms is displayed. In the last row, $n^{*}= \max\{n_{1}, \dots, n_{d}\}$.}
\end{table}

\subsection{VQLS Application}%%%%%%%%%%%%%%%%%%%%%%%%%%%%%%%%%%%%%%%%%%%%%%%%%%

As a demonstration of potential application, we consider the costs of using our methods in conjunction with the VQLS algorithm. We begin with a simple estimate for the number of gates required to execute the Hadamard Test circuits associated with our decompositions and the local cost function proposed in \cite{bravo2023variational} and defined in equation \eqref{local_cost}.  

\begin{lemma}{(VQLS Cost)}\thlabel{lmVQLS}
    The gate cost of a Hadamard Test circuit for estimating the value of a $\beta_{ll'}$ or $\delta^{(j)}_{ll'}$ term corresponding to the proposed LCU of either $A_{n,1}, A_{n,2}$ or $A_{n,3}$ is at most $6$-Hadamard, $O(n)$ $X$, $CNOT$, and Toffoli gates, plus twice the gate costs of $V(\theta)$ and $U_{b}$, provided two ancilla qubits are available. 
\end{lemma}

\begin{proof}
    To begin, let $R_{1,l}$ be an LCU decomposition term associated with the Dirichlet Laplacian $A_{n,1}$. Consider the controlled application of such an $R_{1,l}$ term, which is required to estimate the values of $\beta_{ll'}$ and $\delta^{(j)}_{ll'}$, defined in \eqref{beta_def} and \eqref{delta_def}, using the Hadamard Test circuits shown in Figures \ref{fig_1} and \ref{fig_2}. Note that we set aside one ancilla qubit to be used in the usual way for the Hadamard test and reserve the other to aid in the decomposition of $R_{1,l}$ terms. It follows from equation \eqref{DL3} that providing an upper bound on the cost of implementing a controlled $S_{n}$ or $C_{n}^{-}$ gate will be sufficient to establish the desired estimate for $A_{n,1}$.  
    
    Recall that an $S_{n}$ gate can be constructed using $O(n)$ $X$, $CNOT$, and Toffoli gates in our setting since an ancilla qubit is available \cite{gidney2015constructing}. Therefore, a controlled $S_{n}$ gate will require at most $O(n)$ $CNOT$, Toffoli, and $CNOT(C, t)$ gates, where $|C| =3$. Note that a multicontrol $NOT$ gate with $3$ controls can be written as a product of $4$-Toffoli gates given, as is the case here, access to a borrowable ancilla qubit. On the other hand, a $C_{n}^{-}$ gate can be represented as a product of $3$-$X$, $2(n-1)$-$CNOT$, and one $C_{n-1}(Z)$ gates by \eqref{basic_decomp_1}. Upon adding the additional control, and factoring the resultant multicontrol $Z$ gate into $2$-Hadamard gates and a single $CNOT(C,n)$ gate with $|C|=n-1$, we find that the gate cost associated with $C_{n}^{-}$ is bounded above by $2$-$H$ and $O(n)$ $X$, $CNOT$, and Toffoli gates. Note that we have relied on \cite{barenco1995elementary} once again to estimate the cost of the $CNOT(C, n)$ term (see the proof of \thref{lm1DLCU} above). In either case, we have established that a controlled $R_{1,l}$ gate associated with $A_{n,1}$ requires at most  $2$-$H$ and $O(n)$ $X$, $CNOT$, and Toffoli gates. Moreover, an analogous argument will show that this estimate remains valid for $R_{3,l}$ belonging to the LCU decomposition of $A_{n,3}$ (an additional $2$-Hadamard gates may be required). We omit the details here for brevity. Observe also that each unitary $R_{2,l}$ of the periodic LCU given in \eqref{periodic_LCU} is equal to some $R_{1,l}$ term, and the estimate for the Dirichlet cost remains valid in the periodic case. 
    
    In summary, a controlled $R_{\alpha_{i},l}$ gate, with $\alpha_{i} \in \{1,2,3\}$, can be implemented in this setting using at most $4$-$H$ and $O(n)$ $X$, $CNOT$, and Toffoli gates. Recall the Hadamard test circuits depicted in Figure \ref{fig_1} and Figure \ref{fig_2} and note that the corresponding gate cost will be at most that of a controlled $R_{\alpha_{i},l}$, a controlled $R_{\alpha_{i},l'}^{\dagger}$, a single controlled $Z$, a $V(\theta)$, and two $U_{b}$ gates. The result follows directly from factoring the controlled $Z$ gate into $2$ Hadamard gates and a $CNOT$ and applying the estimate just derived for the gate cost of a controlled $R_{\alpha_{i},l}$ gate, which remains valid for a controlled $R^{\dagger}_{\alpha_{i},l}$ gate.  
\end{proof}

\thref{lmVQLS} can be extended in a straightforward manner to cover the higher-dimensional operators $L^{d}_{\bold{n},\alpha}$. In particular, a naive analysis of the gate cost and circuit depth for a Hadamard Test circuit in the $d$-dimensional case provides an upper bound on the order $O(n^{*})$, where $n^{*} = \max\{n_{1}, \dots, n_{d}\}$. However, it is important to highlight that many additional borrowable qubits are available in this case and sharper estimates are surely obtainable. 

\subsection{Comparison with Existing Methods}%%%%%%%%%%%%%%%%%%%%%%%%%%%%%%%%%%%%%%%%%%%%%%%%%

In this subsection, we demonstrate the advantage of utilizing our decomposition within the VQLS framework over other existing methods. For concreteness, we limit our analysis to $1$-dimensional (inhomogeneous) Dirichlet problems on the unit interval. Access to an ancilla qubit to aid in circuit implementation is assumed throughout. We find that our representation of $A_{n,1}$ given in \eqref{basic_decomp_1} requires at least $n^{2}$ fewer Toffoli gate executions per iteration of the VQLS algorithm (using a local cost function) when compared to the leading Dirichlet LCU found in the literature \cite{kharazi2024explicit}, provided $n$ is sufficiently large and the well-known decomposition of $T_{k}$ and $S_{k}$ gates found in \cite{barenco1995elementary} and \cite{gidney2015constructing}, respectively, are adopted. This reduction in cost can become quite significant as $n$ increases, especially when one considers that the number of iterations required for successful optimization is expected to increase with $n$.

The Poisson problem is fundamental and has naturally attracted the attention of researchers interested in LCU based quantum differential equation solvers. A $3$ term LCU decomposition analogous to \eqref{periodic_LCU} was introduced in \cite{childs2021high}. These authors use a Quantum Fourier Transform technique to tackle both periodic and homogeneous boundary conditions. They also developed pseudo-spectral methods to treat inhomogeneous conditions, but this involves oracles for the boundary values and a separation of the operator according to interior and boundary points; this method is fundamentally different from ours and does not seem to yield an explicit LCU for direct comparison. Around the same time, a different research group proposed a separate $3$ term LCU decomposition for $A_{n,2}$ was proposed in \cite{sato2021variational}. The authors extended this construction to produce a $4$ term decomposition for $A_{n,1}$, but their representation includes non-unitary operators (and is not technically in the LCU form). The authors of \cite{bae2024hardware} modified this approach to avoid implementation of $S_{n}$ gates and obtained a decomposition that includes $n$ non-unitary terms. The non-unitary terms of \cite{sato2021variational,bae2024hardware} are a nontrivial technical obstacle. Nevertheless, these authors proposed a method to calculate the required expectation values specifically tailored to an appropriate global cost function. However, this is a serious restriction since global cost functions are known to be more susceptible to barren plateaus than local cost functions \cite{bravo2023variational}. The non-unitary decompositions described in the articles referenced above might lend themselves to optimization via local cost functions using a block-encoding strategy along the lines of \cite{gnanasekaran2024efficient, gnanasekaran2025efficient}, but a detailed investigation and comparison with such an approach is outside the scope of this paper. Our interest in explicit LCU decompositions for inhomogeneous Dirichlet problems and optimization via local cost functions motivates us to look elsewhere for an in-depth comparison. 

Next, we compare our method with the one developed in \cite{kharazi2024explicit}. We note that those authors created their decompositions with the intention of block-encoding the Laplacian, rather than combining it with the VQLS algorithm. Nevertheless, they also obtain LCU decompositions for Poisson problems and we believe a comparison is warranted. Their LCU for the $1$-dimensional Dirichlet case can be written in the following form: 
\begin{equation}\label{explicit_khar}
    A_{n,1} = 2I^{\otimes n}+\frac{1}{2}S_{n}+\frac{1}{2}S^{\dagger}_{n}+ \frac{1}{2}S_{n}\cdot C_{n}(Z)+\frac{1}{2}C_{n}(Z)\cdot S_{n}^{\dagger} \, . 
\end{equation}
Note that equations \eqref{DL3} and \eqref{explicit_khar} have the same number of terms (we are unaware of a Dirichlet LCU with fewer than $5$ terms). Therefore, we ignore the costs of $U_{b}$, $V(\theta)$, Hadamard, and controlled $Z$ gates incurred from the Hadamard Test in the following analysis, since they are expected contribute equally to the VQLS resource requirements of either decomposition strategy. Moreover, we assume that all circuits are real-valued and exploit the computational reductions described in Section \ref{vqls_subsec}.

For convenience, let $|U|$ denote the gate cost of implementing a generic quantum gate $U$. This notion will be made more precise as we narrow the analysis. We consider an $n+2$ qubit system where the $0\textsuperscript{th}$ qubit will be used as the ancilla to be measured in the Hadamard Test circuits and the $(n+1)\textsuperscript{th}$ qubit will be used as the ancilla that aids in our gate decompositions. Also, let $C(0; S_{n})$ denote a controlled $S_{n}$ gate with the $0\textsuperscript{th}$ qubit controlled. The cost of executing the Hadamard Test circuits as depicted in Figure \ref{fig_2} for $\delta^{(j)}_{ll'}$, for all $l,l' \in \{1, \dots, 5\}$ and a fixed $j \in \{1,\dots, n\}$, with the LCU presented in \cite{kharazi2024explicit} is 
\begin{equation}\label{K_cost}
    G_{1} = 24|C(0;S_{n})|+12|C_{n+1}(Z)| \, ,  
\end{equation}
which can be directly computed using \eqref{explicit_khar} and the definition of $\delta^{(j)}_{ll'}$ from \eqref{delta_def}. On the other hand, the analogous cost for each such term in our decomposition is 
\begin{equation}\label{M_cost}
    G_{2} = 12|C(0;S_{n})|+6(n+3)|CNOT| + 12(n-1)|T_{3}| +6|C_{n}(Z)|\, , 
\end{equation}
which is found in a similar fashion by using \eqref{delta_def} and \eqref{DL3}. There are many approaches to decompose a multicontrol $Z$ gate, and no specific choice seems to be made in \cite{kharazi2024explicit} (they do assume an implementation of such gates with $O(n)$ Toffoli cost). We are left to choose an implementation technique with an $O(n)$ Toffoli cost to begin comparing \eqref{K_cost} and \eqref{M_cost}. It follows from Lemma 7.4 of \cite{barenco1995elementary} that a $C_{n+1}(Z)$ gate can be realized as the product of $2$-Hadamard and $8(n-4)$-Toffoli gates with access to an ancilla qubit. Likewise, $C_{n}(Z)$ can be implemented with $2$-Hadamard and $8(n-5)$-Toffoli gates by the same method. We adopt these representations to estimate the gate costs of multicontrol $Z$ gates. Subtracting \eqref{M_cost} from \eqref{K_cost} then yields the excess gate count incurred by \eqref{explicit_khar} relative to \eqref{DL3} for $\delta^{(j)}_{ll'}$, with $j \in \{1, \dots, n\}$: 
\begin{equation} \label{C_D}
    G_{1}-G_{2} =12 |C(0;S_{n})| + 12(3n-11)|T_{3}| -6(n+3)|CNOT| +12|H|  \, . 
\end{equation}
In order to make sense of \eqref{C_D}, we must give a more concrete interpretation of  $|T_{3}|$ and $|C(0;S_{n})|$. To this end, we make the assumption that $|T_{3}| \geq |CNOT|$, which is reasonable since a $T_{3}$ gate requires at least $6$-$CNOT$ gates to implement if decomposed into $CNOT$ and one qubit gates \cite{shende2008cnot}. Additionally, an $S_{n}$ gate will be implemented here according to \cite{gidney2015constructing}, because this approach is taken in \cite{kharazi2024explicit}. Consequently, it will be assumed that $|S_{n}|\geq n|T_{3}|$. Moreover, as noted in \cite{gidney2015constructing}, $|C(0; S_{n})| = |X| + |S_{n+1}| \geq (n+1)|T_{3}|$. 

It follows from \eqref{C_D} and the reasoning at the end of the preceding paragraph that 
\begin{equation*}
    G_{1}-G_{2} \geq n|T_{3}| \, ,
\end{equation*}
whenever $n\geq 4$. Recall that this excess $T_{3}$ count is incurred for each collection of Hadamard test circuits corresponding to a fixed $j \in \{1, \dots, n\}$. Hence, the LCU presented in \cite{kharazi2024explicit} requires at least $n^{2}$ additional Toffoli gates per iteration of the VQLS algorithm to calculate $\{\delta^{(j)}_{ll'}\}_{j=1}^{n}$. Moreover, this is enough to establish our claim since the constant number of Hadamard Test circuits corresponding to the terms omitted from the preceding analysis, namely the $\beta_{ll'}$ and $\delta^{(n+1)}_{ll'}$ terms, require at most $O(n)$ Toffoli, $CNOT$, and Pauli gates.

\section{Proofs of the main results} \label{proof_sect}%%%%%%%%%%%%%%%%%%%%%%%%%%%%%%%%%%%%%%%

Most of the work required to prove our main results has been presented in the previous sections. Section \ref{Decomp_sect} establishes our LCU decompositions and resource estimates were provided in Section \ref{complexity_sect}. It remains to tie these results together. We begin with the proof of \thref{thm1}. 

\begin{proof}{(\thref{thm1})}
    An LCU decomposition for $A_{n,1}$ in form of equation \eqref{thm1_eq} follows from equation \eqref{DL3}, which was justified by the work in Section \ref{dirichlet_decomp_fund}. It is easy to see that each term in the decomposition is unitary. For concreteness, we label $c_{1,l}$ and $R_{1,l}$ according to \eqref{dirichlet_decomp_fund} in order of appearance so that, for example, $c_{1,1}=-2$, $R_{1,l} = I^{\otimes n}$, $c_{1,2} =1$, $R_{1,2} = S_{n}$, and so on. \thref{lm1DLCU} establishes the desired gate cost estimates for the LCU terms. 
\end{proof}
The Neumann/Robin case is handled next. 
\begin{proof}{(\thref{thm2})}
    Equation \eqref{NR_LCU} provides an LCU decomposition for $A_{n,3}$ in the form of \eqref{thm2_eq}. Note that \eqref{NR_LCU} contains at most $10$ terms, which are all evidently unitary, although some of the coefficients may vanish depending on the problem-specific values of $a_{0}$, $a_{1}$, and $h_{3}$. We label $c_{3,l}$ and $R_{3,l}$ in order of appearance in equation \eqref{NR_LCU}. The claimed gate cost estimates are validated by \thref{lm1DLCU}.
\end{proof}
The higher-dimensional result is handled analogously. 
\begin{proof}{(\thref{thm3})}
    An appropriate LCU decomposition is obtained from the second line of \eqref{higher_dim_LCU}. Note that we take the choice of labeling for $c^{i}_{\alpha_{i},j}$ and $R^{i}_{\alpha_{i},j}$ to be consistent with those introduced in the proofs for \thref{thm1} and \thref{thm2}, and that each term of the proposed decomposition is clearly unitary. Gate cost estimates follow directly from \thref{HD_lem}. 
\end{proof}

\section{Extensions}\label{ext_sect}%%%%%%%%%%%%%%%%%%%%%%%%%%%%%%%%%%%%%%%%%%%%%%%%%%%%%%%%%%%%%
The decompositions presented in Section \ref{Decomp_sect} are restricted to problems whose corresponding differential operator is simply the Laplacian. In this section, we demonstrate efficient LCU decompositions for some differential operators that include first-order derivatives. We restrict our attention to $1$-dimensional Dirichlet problems, but our methods could be extended in a straightforward manner to higher-dimensions and alternate boundary conditions. 

\subsection{Constant Coefficient First-Order Terms}%%%%%%%%%%%%%%%%%%%%%%%%%%%%%%%%%%%%%%%%%%%%%%%%%%%%%%%%%%
Consider the following problem: 
\begin{empheq}[left=\empheqlbrace]{align}
\label{basic_grad}
\begin{split}
         \frac{d^{2}}{dx^{2}}u+p \frac{d}{dx}u&=f \qquad   \;\;\text{in} \qquad  (0,1)\\  
                                               u(0)& = a \\ 
                                               u(1)&=b \, . 
\end{split}
\end{empheq}
Let us adopt the notation used in Section \ref{1d_Dir_sub_sect}. We can approximate $\frac{d}{dx}$ using a second-order accurate central difference scheme:
\begin{equation}\label{cont_grad}
    D_{n}\bold{u} \coloneqq \frac{1}{2h_{1}}\begin{pmatrix}
            0 & 1  & 0 & \cdots &  & 0 \\
            -1  & 0 & 1 & \ddots &  & \vdots \\ 
            0  & -1 & 0 & 1 &  &  \\
        \vdots &    & \ddots  & \ddots   &  \ddots &  0 \\
              &       &   \ddots &  -1 & 0 &  1 \\
            0  &       &  \cdots  &   0& -1 & 0 \\
            \end{pmatrix}
            \bold{u} 
            = \bold{u}'(\bold{x})+  \frac{1}{2h_{1}}
            \begin{pmatrix}
             -a \\
             0 \\
              \vdots \\
               \\
               0\\
              b \\
            \end{pmatrix} + O(h_{1}^{2}) \, . 
\end{equation}
Note that the matrix $D_{n}$ in \eqref{cont_grad} is simply expressed as 
\begin{equation}\label{basic_grad_exp}
    D_{n} = \frac{1}{2h_{1}}\left(S_{n}^{\dagger} - S_{n} +\left(I^{\otimes n-1}\otimes Z\right) \cdot C_{n}\right)\, , 
\end{equation}
where $C_{n}$ is defined in \eqref{Cpm}. Equations \eqref{DL3} and \eqref{cont_grad} can be combined to obtain the following second-order accurate approximation of \eqref{basic_grad}:
\begin{equation} \label{L_alpha_Dir}
    \left(L_{n,1}+p D_{n}\right)\bold{u} = \bold{f} + \tilde{B}_{1} \, , 
\end{equation}
where 
\begin{equation*}
    \tilde{B}_{1} = \left(\frac{a(ph_{1}-2)}{2h_{1}^{2}}, 0, \dots, 0,\frac{-b(ph_{1}+2)}{2h_{1}^{2}}\right)^{T}
\end{equation*}
Equation \eqref{L_alpha_Dir} can be normalized by a factor of $h_{1}^{2}$ to yield
\begin{equation}\label{L_alpha_dirch_2}
     \left(A_{n,1}+ph_{1}^{2} D_{n}\right)\bold{u} = h_{1}^{2}\bold{f} + h_{1}^{2}\tilde{B}_{1} \, , 
\end{equation}
and the matrix on the left hand side of \eqref{L_alpha_dirch_2} can be expanded using \eqref{Cnd}, \eqref{DL3}, and \eqref{basic_grad_exp} to read 
\begin{align}\label{intm_1}
\begin{split}
   A_{n,1}+ph_{1}^{2} D_{n}= &-2I^{\otimes n}+\frac{(2-p h_{1})}{2}S_{n}+\frac{(2+p h_{1})}{2}S_{n}^{\dagger}- \frac{1}{2}X^{\otimes n}-\frac{1}{2}C^{-}_{n} \\ 
    &+\frac{h_{1}p}{4}\left(I^{\otimes n-1}\otimes Z\right)\cdot X^{\otimes n} +\frac{h_{1}p}{4}\left(I^{\otimes n-1}\otimes Z\right)\cdot C^{-}_{n}\, . 
\end{split}
\end{align}
Hence, \eqref{intm_1} is an LCU for a second-order accurate numerical scheme that encodes a Poisson problem with a constant coefficient first order derivative term. The following theorem summarizes our findings. 

\begin{theorem} \thlabel{thm4}
    The matrix $A_{n,1}+ph_{1}D_{n}$ admits a decomposition of the form
    \begin{equation*} \label{thm4_eq}
        A_{n,1} +ph^{2}_{1}D_{n}= \sum_{l=1}^{7}c_{4,l}R_{4,l}\, , 
    \end{equation*}
    where $c_{4,l} \in \mathbb{R}$ and each $R_{4,l}$ is unitary. Moreover, if a single ancilla qubit is available, then a given $R_{4,l}$ may be implemented using at most a single $Z$ gate, $2$-Hadamard gates and $O(n)$ Toffoli, $X$, and $CNOT$ gates.
\end{theorem}
\begin{proof}
    Equation \eqref{intm_1} provides an LCU decomposition of $A_{n,1} +ph_{1}D_{n}$ with $7$ terms. The claimed complexity of implementing a single term from this decomposition follows immediately from the proof of \thref{lm1DLCU}. 
\end{proof}
\subsection{Polynomial Coefficients}\label{PC_sub} %%%%%%%%%%%%%%%%%%%%%%%%%%%%%%%%%%%%%%%%%%%%%%%%%%%%%%%%%%%%%%%%%%
In this subsection, we present a method to derive LCU decompositions for Poisson problems augmented with polynomial-coefficient first-order derivative terms. In particular, suppose that $p(x) = a_{0}+a_{1}x+\cdots +a_{k}x^{k}$ and consider the following problem:
\begin{empheq}[left=\empheqlbrace]{align}
\label{poly_grad}
\begin{split}
    \frac{d^{2}}{dx^{2}}u+p(x) \frac{d}{dx}u&=f \qquad   \;\;\text{in} \qquad  (0,1)\\  
                    u(0)& = a \\ 
                    u(1)&=b \, . 
\end{split}
\end{empheq}
We will build off the work of the previous subsection to derive a discretization for \eqref{poly_grad}. Recall that $x_{i} = ih_{1}$ and let
\begin{equation*}
  M(p(x)) \coloneqq  \begin{pmatrix}
        p(x_{1}) &  &  &  & \\
                  & p(x_{2}) & & & \\
                  & & &\ddots & \\ 
                  & & & & p(x_{2^{n}})
    \end{pmatrix} \, . 
\end{equation*}
It can be shown that 
\begin{equation}\label{FD_poiss_var}
    \left(L_{n,1}+M(p(x))D_{n}\right)\bold{u} = \bold{f} +\tilde{B}_{1} \, , 
\end{equation}
where we have introduced the following slight abuse of notation
\begin{equation*}
    \tilde{B}_{1} = \left(\frac{a(p(x_{1})h_{1}-2)}{2h_{1}^{2}}, 0, \dots, 0,\frac{-b(p(x_{2^{n}})h_{1}+2)}{2h_{1}^{2}}\right)^{T} \, , 
\end{equation*}
is a second-order accurate discretization of \eqref{poly_grad}.  Equation \eqref{FD_poiss_var} can be normalized by a factor of $h_{1}^{2}$ to get
\begin{equation}\label{L_p_dirch_2}
     \left(A_{n,1}+h_{1}^{2} M(p(x))D_{n}\right)\bold{u} = h_{1}^{2}\bold{f} + h_{1}^{2}\tilde{B}_{1} \, . 
\end{equation}

We are left now to decompose $M(p(x))$. To this end, let $c_{j} = -\frac{2^{n}}{2^{n}+1}2^{-j}$, $Z_{j} = I^{\otimes j-1}\otimes Z \otimes I^{\otimes n-j}$, and $\mathcal{A}_{k} = \{ \{i_{1}, \dots, i_{k}\} \; \; | \; \; 1\leq i_{1}<i_{2} < \cdots < i_{k}\leq n\}$. For $\alpha = \{i_{1}, \dots, i_{m}\} \in \mathcal{A}_{m}$, we introduce the definitions  
\begin{align*}
    \begin{split}
        c_{\alpha} &\coloneq c_{i_{1}}\cdots c_{i_{m}} \\
        Z_{\alpha} &\coloneq Z_{i_{1}}\cdots Z_{i_{m}} \, . 
    \end{split}
\end{align*}
It is shown in Appendix \ref{Append_formula_poly_dec} that this multi-index notation can be used to write 
\begin{align} \label{MP_final_it}
    \begin{split}
        M(p(x)) &= \left(a_{0}+\sum_{j=1}^{k}\frac{a_{j}c_{0}(j)}{2^{j}}\right) I^{\otimes n} + \sum_{i=1}^{n}\left(\sum_{j=1}^{k}\frac{a_{j}c_{i}(j)}{2^{j}}\right)c_{i}Z_{i} + \cdots \\ 
        &+ \sum_{\alpha \in \mathcal{A}_{r}}\left( \sum_{j=r}^{k}\frac{a_{j}c_{\alpha}(j)}{2^{j}}\right)c_{\alpha}Z_{\alpha} + \cdots + \sum_{\alpha \in \mathcal{A}_{k}}\frac{a_{k}c_{\alpha}(k)}{2^{k}}c_{\alpha}Z_{\alpha} \, . 
    \end{split}
\end{align}
See \eqref{Z_alpha_coeffs} and \eqref{rr_px} for a recursive relation defining the coefficients $c_{\alpha}(j)$, or Appendix \ref{Appendix_C} to view the explicit values corresponding to the cubic case. Equation \eqref{MP_final_it} contains $\sum_{j=0}^{k} {n \choose j}$ unitary terms of the form $Z_{\alpha}$, where $|\alpha| \in \{0,\dots, k\}$. It is well known that no closed-form formula for $\sum_{j=0}^{k} {n \choose j}$ exists, but estimates for some for some ranges of $k$ and $n$ are easily derived. For example, it can be shown that 
\begin{equation} \label{bound_on_poly}
    \sum_{j=0}^{k} {n \choose j} \leq \left(\frac{ne}{k}\right)^{k}  
\end{equation}
whenever $k<n/2$. Hence, equation \eqref{MP_final_it} contains $O(n^{k})$ $Z_{\alpha}$ terms for fixed $k$. 

It follows that $A_{n,1}+h_{1}^{2} M(p(x))D_{n}$ can be decomposed into $O(n^{k})$ unitary terms. We record our findings in the following theorem. 
\begin{theorem} \thlabel{thm5}
    Let $k <\frac{n}{2}$ and $p(x) =a_{0}+a_{1}x+\cdots +a_{k}x^{k}$. Then, the matrix $A_{n,1}+M(p(x))h^{2}_{1}D_{n}$ admits a decomposition of the form
    \begin{equation} \label{thm5_eq}
        A_{n,1} +M(p(x))h_{1}D_{n}= \sum_{l=1}^{L}c_{5,l}R_{5,l}\, , 
    \end{equation}
    where $L$ is $O(n^{k})$, $c_{5,l} \in \mathbb{R}$ and each $R_{5,l}$ is unitary. Moreover, if a single ancilla qubit is available, then a given $R_{5,l}$ may be implemented using at most $2$-Hadamard gates, $(k+1)$-$Z$ gates, and $O(n)$ Toffoli, $X$, and $CNOT$ gates.
\end{theorem}
\begin{proof}
    It follows from \eqref{MP_final_it} and \eqref{bound_on_poly} that $M(p(x))$ can be written as the sum of $O(n^{k})$ unitary terms, each of which contains at most $k$-$Z$ gates. An LCU of the desired form can be obtained by first substituting $A_{n,1}$, $D_{n}$, and $M(p(x))$ for \eqref{DL3}, \eqref{basic_grad_exp}, and \eqref{MP_final_it}, respectively (note that $C_{n}$ must also be substituted using \eqref{Cnd}), and then distributing $M(p(x))$ terms against $D_{n}$ terms. In particular, one obtains an LCU with $O(n^{k})$ unitary terms. The claimed upper bound on the gate cost of implementing a single term from this decomposition follows from the estimates derived in \thref{lm1DLCU}, since the terms in this setting contain, in the worst case, an additional $k$-$Z$ gates. 
\end{proof}

We emphasize that \eqref{thm5_eq} is an exact representation of the discrete differential operator on the left hand side of \eqref{L_p_dirch_2}, and reductions in cost are likely possible with some suitable truncation that removes the least significant terms of $M(p(x))$. Moreover, \eqref{thm5_eq} can be generalized to treat problems of the form $u_{xx}+g(x)u_{x}$, where $g(x)$ is well-suited to polynomial approximation. 

\section{Discussion}\label{disc_sect}%%%%%%%%%%%%%%%%%%%%%%%%%%%%%%%%%%%%%%%%%%%%%%%%%%%%%%%%%%%%%%%%%%%%%%%
Partial differential equations are among the most important objects in applied mathematics. Quantum computing may help meet the constant demand for increased accuracy and efficiency in the numerical methods used to solve them. In this work, we present novel LCU decompositions for a variety of discrete Poisson problems. The formulations considered allow for rectangular domains in arbitrary dimension $d$ and periodic, Dirichlet, Neumann, Robin, or mixed boundary conditions. Our decompositions are comprised of at most $O(d)$ unitary terms which can be individually implemented using at most a single $Z$ gate, $2$-Hadamard gates, and $O(n)$ Toffoli, $X$, and $CNOT$ gates if a single ancilla qubit is available (no ancilla are required in the $d>2$ setting). We emphasize that the size of the decompositions are independent of the number of grid points and scale linearly with the spatial dimension. We highlight the advantages of our method by giving cost estimates of our $1$-dimensional Dirichlet LCU within the VQLS algorithm. We demonstrate a quadratic improvement over other leading decompositions in terms of total Toffoli gates required per iteration. Our methods could be combined with a wide variety of VQAs or quantum linear solvers that utilize an LCU strategy and we believe comparable advantages can be realized in many cases. Furthermore, our decompositions could be incorporated as a subroutine for LCU based quantum PDE solvers whose underlying equation involves the Laplacian.   

Moreover, we provide a method to generate LCU decompositions for Poisson problems augmented with polynomial-coefficient first-order derivative terms. We achieve this by introducing a technique to encode polynomial data on $n$ qubit systems using at most $O(n^{k})$ Pauli terms, where $k<n/2$ is the degree of the polynomial under consideration. In the $1$-dimensional Dirichlet setting, we obtain an LCU with $O(n^{k})$ terms which can be implemented with at most $2$ Hadamard, $(k+1)$ $Z$, and $O(n)$ Toffoli, $X$, and $CNOT$ gates. We believe this polynomial encoding method could be more widely useful. Polynomial data loading is an active area of research and alternate strategies exist \cite{gonzalez2024efficient, marin2023quantum, welch2014efficient}. While these methods may be preferable for some applications, they make use of more sophisticated machinery and incur serious overhead in terms of rotation gates, $CNOT$ gates, and additional ancilla. Our strategy has the advantage of requiring only Pauli-$Z$ gates, no ancilla, and retains desirable poly-logarithmic scaling for any fixed polynomial degree.  

Several promising avenues for future research are touched upon throughout this article. The most straightforward extensions include deriving an LCU method for time-dependent problems such as the heat or wave equations, generalizing the $1$-dimensional polynomial-coefficient decomposition to treat higher-dimensional problems posed on box-type domains, and investigating the performance of our methods when combined with fault tolerant algorithms based on block-encoding in the spirit of \cite{childs2017quantum}. Another natural topic for further investigation is an in-depth case study of VQLS performance; this would require one to adopt strategies to encode boundary values and the forcing function, design an appropriate ansatz, and perform solution readout. Finally, we mention that our decompositions might be adaptable to penalized FDM schemes. This could be a step towards tackling the critical problem of designing quantum PDE solvers for problems posed on non-rectangular domains, and we plan to address it in future research.   
\newpage

\bibliographystyle{unsrt}
\bibliography{references}

\appendix

\section{LCU Decompositions for $\mathbb{T}$ Terms} \label{Apen_A}
Several useful results regarding tensor products of elements belonging to $\mathbb{T}$ are derived in this appendix. First, we produce a decomposition for $\tau_{0}^{\otimes n}$. Recall the definition of $C_{n}$ given in \eqref{Cpm}. It is clear that 
\begin{equation*}
    \tau_{0}^{\otimes n} = \begin{pmatrix}
        1 & & & \\
          &0& & \\ 
          & & \ddots& \\ 
          & &      & 0
    \end{pmatrix} = \frac{1}{2}\left(X^{\otimes n}\cdot C_{n}+X^{\otimes n}\cdot C_{n}\cdot \left( I^{\otimes n-1}\otimes Z\right)\right) \, . 
\end{equation*}
The alternate expression for $C_{n}$ given by equation \eqref{Cnd} can now be used to show 
\begin{equation} \label{tau_zero}
        \tau_{0}^{\otimes n} = \frac{1}{4}\left(X^{\otimes n}\cdot C^{-}_{n}+I^{\otimes n}+X^{\otimes n}\cdot C^{-}_{n}\cdot \left(I^{\otimes n-1}\otimes Z\right) + I^{\otimes n-1}\otimes Z\right) \, .
\end{equation}
One can derive expressions for $\tau_{j}^{\otimes n}$, where $j \in \{2,3,4\}$, using \eqref{tau_zero} and the relations
\begin{align}\label{tau_0_rels}
    \begin{split}
        X\cdot \tau_{0} & = \tau_{2} \\ 
        \tau_{0} \cdot X & = \tau_{1} \\ 
        X\cdot \tau_{0} \cdot X & = \tau_{3} \, . 
    \end{split}
\end{align}
We summarize the results below:
\begin{align} \label{tau_3}
\begin{split}
        \tau_{0}^{\otimes n} &= \frac{1}{4}\left(
    I^{\otimes n}+ I^{\otimes n-1}\otimes Z + X^{\otimes n} \cdot C_{n}^{-}+X^{\otimes n}\cdot C_{n}^{-}\cdot I^{\otimes n-1}\otimes Z
    \right) \\
        \tau_{1}^{\otimes n} &= \frac{1}{4}\left(C^{-}_{n}+X^{\otimes n}-C^{-}_{n}\cdot \left(I^{\otimes n-1}\otimes Z\right) + \left(I^{\otimes n-1}\otimes Z\right)\cdot X^{\otimes n}\right) \\
        \tau_{2}^{\otimes n} &= \frac{1}{4}\left(C^{-}_{n}+X^{\otimes n}+C^{-}_{n}\cdot \left(I^{\otimes n-1}\otimes Z\right) +X^{\otimes n}\cdot \left( I^{\otimes n-1}\otimes Z\right)\right) \\
        \tau_{3}^{\otimes n} &= \frac{1}{4}\left(X^{\otimes n}\cdot C^{-}_{n}+I^{\otimes n}-X^{\otimes n}\cdot C^{-}_{n}\cdot \left(I^{\otimes n-1}\otimes Z\right) - I^{\otimes n-1}\otimes Z\right) \, . 
\end{split}
\end{align}

Next, we record an efficient decomposition for tensor products of $\tau$ terms. If $\tau_{r_{j}} \in \mathbb{T}$, for $j \in \{1,\dots, n\}$ and $r_{j} \in \{0,1,2,3\}$, then equation \eqref{tau_0_rels} can be used to show
\begin{equation} \label{tau_4}
    \bigotimes_{j=1}^{n} \tau_{r_{j}} = \left(\bigotimes_{j=1}^{n}X^{\nu_{r_{j}}}\right)\cdot\left(\bigotimes_{j=1}^{n}\tau_{0}\right)\cdot \left(\bigotimes_{j=1}^{n}X^{\sigma_{r_{j}}}\right) \, , 
\end{equation}
where 
\begin{equation} \label{tau_5}
    (\nu_{r_{j}}, \sigma_{r_{j}}) \coloneqq \begin{cases}
        (0,0), \qquad \text{for} \qquad \tau_{r_{j}} = \tau_{0}\\ 
        (0,1), \qquad \text{for} \qquad \tau_{r_{j}} = \tau_{1}\\ 
        (1,0), \qquad \text{for} \qquad \tau_{r_{j}} = \tau_{2}\\ 
        (1,1), \qquad \text{for} \qquad \tau_{r_{j}} = \tau_{3} \, . 
    \end{cases}
\end{equation}
Finally, equations \eqref{tau_zero}, \eqref{tau_4}, and \eqref{tau_5} can be used in conjunction to yield 
\begin{equation} \label{tau_6}
    \bigotimes_{j=1}^{n} \tau_{r_{j}} = \left(\bigotimes_{j=1}^{n}X^{\nu_{r_{j}}}\right)\cdot\left(
    I^{\otimes n}+ I^{\otimes n-1}\otimes Z + X^{\otimes n} \cdot C_{n}^{-}+X^{\otimes n}\cdot C_{n}^{-}\cdot I^{\otimes n-1}\otimes Z
    \right)\cdot \left(\bigotimes_{j=1}^{n}X^{\sigma_{r_{j}}}\right) \, . 
\end{equation}
Note that \eqref{tau_6} provides a $4$ term LCU decomposition for a generic tensor product of elements from $\mathbb{T}$. Furthermore, it follows directly from the resource estimates of $C_{n}^{-}$ provided in \thref{lm1DLCU} that any term from \eqref{tau_6} can be implemented using at most $2$ Hadamard, a single $Z$, and $O(n)$ Toffoli, $X$, and $CNOT$ gates. 

\section{Example $2$-dimensional problem with mixed boundary conditions} \label{Append_B}
Let $\Omega = (0,1)\times(0,1)$. Consider the following mixed Dirichlet--Robin problem: 
\begin{empheq}[left=\empheqlbrace]{align*}
\label{poiss_ex}
\begin{split}
        \Delta u&=f \qquad \qquad  \;\;\text{in} \; \; \; \Omega\\  
         u(0,y) &=g_{0}(y)   \\
         u(1,y) & = g_{1}(y) \\ 
         u_{y}(x,0)+a_{0}u(x,0) &=b_{0}(x) \\ 
         u_{y}(x,1) +a_{1}u(x,1) & = b_{1}(x) \, . 
\end{split}
\end{empheq}
We demonstrated how to encode the interior and boundary differential operators in Section \ref{Decomp_sect}. In particular, if we take $\bold{n} = (n_{1},n_{2})$ and $\alpha = (1,3)$, then 
\begin{equation*}
    L^{2}_{\bold{n},\alpha} = I^{\otimes n_{2}}\otimes L_{n_{1},1}+L_{n_{2},3}\otimes I^{\otimes n_{1}}
\end{equation*}
approximates $\Delta$ inside $\Omega$ and $\partial_{y}+a_{i}$ on $y=i$, where $i \in \{0,1\}$. However, we did not give a complete description in Section \ref{Decomp_sect} for encoding the inhomogeneous boundary data; the relevant functions in the present setting are $g_{i}(y)$ and $b_{i}(x)$ for $i \in \{0,1\}$. As mentioned before, the essential task is to encode the discretized values of $g_{i}, b_{i}$ according to their corresponding grid positions, scale them appropriately, and absorb the resulting values into $f(x_{i}, y_{j})$. Recall that our numerical scheme for Dirichlet boundary conditions excludes the endpoints of an interval by incorporating the prescribed values into the equations for the first and last interior points, whereas our Robin/Neumann scheme includes endpoints. The labeling described in \eqref{Higher_D_Sub} gives 
\begin{equation*}
    \bold{u} = (u_{1,1}, \dots, u_{N_{1},1}, u_{1,2}, \dots , u_{N_{1},2}, \dots, u_{N_{1},N_{2}})^{T}\, , 
\end{equation*}
where $u_{i,j} = u(x_{i}, y_{j})$. It follows that $g_{0}(y_{j})$ is associated with the $\left((j-1)N_{1}+1\right)\textsuperscript{th}$ row, for $j \in \{1, \dots, N_{2}\}$, of $L^{2}_{\bold{n},\bold{\alpha}}$. Likewise, $g_{1}(y_{j})$ corresponds to the $jN_{1}\textsuperscript{th}$ row of $L^{2}_{\bold{n},\bold{\alpha}}$. On the other hand, the values of $b_{0}(x_{j})$ and $b_{1}(x_{j})$ are associated with the first $N_{1}$ and last $N_{2}$ rows of $L^{2}_{\bold{n},\bold{\alpha}}$, respectively. We can now state the full discretized Poisson problem:
\begin{equation}\label{ex_eq}
    L^{2}_{\bold{n},\alpha}\bold{u} = \bold{f}+B^{2}_{\bold{n},\alpha} \, , 
\end{equation}
where $\bold{f} = (f(x_{1},y_{1}), \dots, f(x_{N_{1}},y_{1}),\dots ,f(x_{N_{1}},y_{N_{2}})^{T}$, and 
\begin{align}\label{ex_force}
\begin{split}
     B^{2}_{\bold{n},\alpha} =& -\frac{1}{h_{1}^{2}}(g_{0}(y_{1}), \dots, g_{0}(y_{N_{2}}))^{T}\otimes \ket{0}^{\otimes n_{1}}-\frac{1}{h_{1}^{2}}(g_{1}(y_{1}), \dots, g_{1}(y_{N_{2}}))^{T}\otimes \ket{1}^{\otimes n_{1}} \\ 
     &+\frac{2}{h_{3}}\ket{0}^{\otimes n_{2}}\otimes (b_{0}(x_{1}), \dots, b_{0}(x_{N_{1}}))^{T}-\frac{2}{h_{3}}\ket{1}^{\otimes n_{2}}\otimes(b_{1}(x_{1}), \dots, b_{1}(x_{N_{1}}))^{T} \, . 
\end{split}
\end{align}
We present equations \eqref{ex_eq} and \eqref{ex_force} as an explicit demonstration of how inhomogeneous boundary data is incorporated into the types of discrete problem considered in this work.  

\section{An LCU Decomposition Strategy for Polynomial Data}%%%%%%%%%%%%%%%%%%%%%%%%%%%%%%%%%%%%%%%%%%%%
\label{Append_formula_poly_dec}
Recall the set-up of Section \ref{ext_sect} in which $M(p(x)) = \textup{diag}(p(x_{1}),\dots, p(x_{2^{n}}))$ is a matrix that encodes the discrete values of $p(x) = a_{0}+a_{1}x+\cdots+a_{k}x^{k}$ on $(0,1)$. In particular, $x_{i} = ih_{1} = \frac{i}{2^{n}+1}$ for $i \in \{1, \dots, 2^{n}\}$. This Appendix is dedicated to deriving an LCU decomposition for $M(p(x))$. To this end, a straightforward calculation shows that 
\begin{equation} \label{basic_mx}
    M(x) = \frac{1}{2}\left(I^{\otimes n} - \frac{2^{n}}{2^{n}+1}\sum_{j=1}^{n}2^{-j}Z_{j}\right) \, ,
\end{equation}
where $Z_{j}$ is defined by $Z_{j} = I^{\otimes j-1}\otimes Z \otimes I^{\otimes n-j}$, and consequently 
\begin{equation}\label{basic_px}
    M(p(x)) = a_{0}I^{\otimes n}+a_{1}M(x)+\cdots + a_{k}M^{k}(x) \, . 
\end{equation}
Equation \eqref{basic_mx} provides us with a simple decomposition for $M(x)$ that can be combined with \eqref{basic_px} to construct efficient decompositions for a class of polynomial data. However, this strategy requires us to compute $M^{r}(x)$ for $r \in \{2,\dots, k\}$. 

For convenience, let us rescale \eqref{basic_mx} as 
\begin{equation}\label{basic_px_2}
    J \coloneq 2M(x) = I^{\otimes n} + \sum_{j=1}^{n}c_{j}Z_{j}\, \qquad c_{j} = -\frac{2^{n}}{2^{n}+1}2^{-j}\, . 
\end{equation}
Suppose that one has written the $k$-th power of $J$ in the following form:
\begin{align} \label{KP_exp_1}
\begin{split}
    J^{k} &= c_{0}(k)I^{\otimes n}+\sum_{i=1}^{n}c_{i}(k)c_{i}Z_{i}+\sum_{i_{1}=1}^{n-1}\sum_{i_{2}=i_{1}+1}^{n}c_{i_{1}i_{2}}(k)c_{i_{1}}c_{i_{2}}Z_{i_{1}}Z_{i_{2}}+\cdots  \\ &+ \sum_{i_{1}=1}^{n-k}\sum_{i_{2}=i_{1}+1}^{n-k+1}\cdots \sum_{i_{k}=i_{k-1}+1}^{n}c_{i_{1}\cdots i_{k}}(k)c_{i_{1}}\cdots c_{i_{k}}Z_{i_{1}}\cdots Z_{i_{k}} \, . 
\end{split}
\end{align}
Note that a representation for $J^{k}$ in the form of \eqref{KP_exp_1} always exists; we have simply collected all products of $k$ or fewer distinct $Z_{i}$ terms. Let $\mathcal{A}_{k} = \{ \{i_{1}, \dots, i_{k}\} \; \; | \; \; 1\leq i_{1}<i_{2} < \cdots < i_{k}\leq n\}$. For $\alpha = \{i_{1}, \dots, i_{m}\} \in \mathcal{A}_{m}$, we introduce the definitions  
\begin{align}\label{Z_alpha_coeffs}
    \begin{split}
        c_{\alpha}(k) &\coloneq c_{i_{1},\dots, i_{m}}(k) \\
        c_{\alpha} &\coloneq c_{i_{1}}\cdots c_{i_{m}} \\
        Z_{\alpha} &\coloneq Z_{i_{1}}\cdots Z_{i_{m}} \, . 
    \end{split}
\end{align}
This multi-index notation can be used to rewrite equation \eqref{KP_exp_1} more compactly as 
\begin{equation}\label{KP_exp_2}
    J^{k} = c_{0}(k)I^{\otimes n} + \sum_{i=1}^{n}c_{i}(k)c_{i}Z_{i}+\sum_{\alpha \in \mathcal{A}_{2}}c_{\alpha}(k)c_{\alpha}Z_{\alpha} + \cdots + \sum_{\alpha \in \mathcal{A}_{k}}c_{\alpha}(k)c_{\alpha}Z_{\alpha} \, . 
\end{equation}
Equation \eqref{basic_px} can now be rewritten with the aid of \eqref{basic_px_2} and \eqref{KP_exp_2} as 
\begin{align} \label{MP_final}
    \begin{split}
        M(p(x)) &= \left(a_{0}+\sum_{j=1}^{k}\frac{a_{j}c_{0}(j)}{2^{j}}\right) I^{\otimes n} + \sum_{i=1}^{n}\left(\sum_{j=1}^{k}\frac{a_{j}c_{i}(j)}{2^{j}}\right)c_{i}Z_{i} + \cdots \\ 
        &+ \sum_{\alpha \in \mathcal{A}_{r}}\left( \sum_{j=r}^{k}\frac{a_{j}c_{\alpha}(j)}{2^{j}}\right)c_{\alpha}Z_{\alpha} + \cdots + \sum_{\alpha \in \mathcal{A}_{k}}\frac{a_{k}c_{\alpha}(k)}{2^{k}}c_{\alpha}Z_{\alpha} \, . 
    \end{split}
\end{align}
If $\alpha \in \mathcal{A}_{m}$, then the coefficients for $J^{k+1}$ can be described using the recursive relations
\begin{align}\label{rr_px}
    \begin{split}
        c_{\alpha}(p+1) &= 0 \qquad \qquad \qquad \qquad \qquad \qquad \qquad \qquad \qquad \;\text{for} \qquad m>p+1 \\ 
        c_{\alpha}(p+1) &= (p+1)! \qquad \qquad\qquad\qquad\qquad\qquad\qquad \;\;\;\;\;\text{for} \qquad m=p+1 \\
        c_{\alpha}(p+1) &= c_{\alpha}(p) + \sum_{l \in \alpha}c_{\alpha\setminus \{l\}}(p)+ \sum_{l \in \alpha^{c}}c_{l}^{2}c_{\alpha \cup \{l\}}(p) \qquad \;\;\text{for} \qquad m \in \{0, \cdots, p\} \, . 
    \end{split}
\end{align}
Equation \eqref{rr_px} is verified with an inductive argument by multiplying the right hand side of \eqref{basic_px_2} against \eqref{KP_exp_2} and collecting all values that contribute to the coefficient of $Z_{\alpha}$. In practice, one can use \eqref{MP_final} and \eqref{rr_px} to produce explicit LCU decompositions of $M(p(x))$. 

\section{Derivation of LCU Coefficients for Low-Degree Polynomials} \label{Appendix_C}%%%%%%%%%%%%%%%%%%%%%%%%%
In this Appendix, we demonstrate an application of the methods described in Appendix \ref{Append_formula_poly_dec} for the case of encoding cubic polynomial data. In particular, we derive an LCU decomposition for $M(p(x))$ using \eqref{MP_final} given $p(x) = a_{0}+a_{1}x+a_{2}x^{2}+a_{3}x^{3}$.  

The values of $c_{\alpha}(k)$ in $J^{k}$ are recorded for $|\alpha|\, , k \leq 3$, since they are required in \eqref{MP_final}. Recall that $J = I^{\otimes n}+\sum_{i=1}^{n}c_{i}Z_{i}$, where $c_{i} = -\frac{2^{n}}{2^{n}+1}2^{-i}$ and $c_{\alpha}(k)$ is a coefficient for $c_{\alpha}Z_{\alpha}$ in the expression for $J^{k}$ (see \eqref{KP_exp_2}). It is clear from \eqref{basic_px_2} that 
\begin{align}\label{exp_c1}
    \begin{split}
        c_{0}(1) &= 1 \\
        c_{i}(1) &= 1 \qquad\qquad\qquad 1\leq i \leq n \, .
    \end{split}
\end{align}
Now, \eqref{rr_px} can be used with \eqref{exp_c1} to find all nonzero $c_{\alpha}(2)$ terms:
\begin{align}\label{exp_c2}
    \begin{split}
        c_{0}(2) &= 1 +\frac{2^{n}-1}{3(2^{n}+1)}\\
        c_{i}(2) &= 2 \qquad \qquad \qquad 1\leq i \leq n \\
        c_{ij}(2) &= 2 \qquad \qquad \qquad 1\leq i < j\leq n \, 
    \end{split}
\end{align}
and again with \eqref{exp_c2} to find: 
\begin{align}\label{exp_c3}
    \begin{split}
        c_{0}(3) &= 1 +\frac{2^{n}-1}{2^{n}+1}\\
        c_{i}(3) &= \frac{2(2^{n+1}+1)}{2^{n}+1}-\frac{2^{2n-2i+1}}{(2^{n}+1)^{2}} \qquad 1\leq i \leq n \\
        c_{ij}(3) &= 6 \qquad\qquad\qquad \qquad\qquad \qquad 1\leq i < j\leq n \\
        c_{ijk}(3) & = 6 \qquad\qquad\qquad \qquad\qquad \qquad 1\leq i < j<k \leq n \, . 
    \end{split}
\end{align}
The recursive relation \eqref{rr_px} could be applied again to determine all values $c_{\alpha}(4)$ and so on. Finally, we use \eqref{MP_final}, \eqref{exp_c1}, \eqref{exp_c2}, and \eqref{exp_c3} to conclude
\begin{equation} \label{final_cub_poly}
    M(p(x)) = A_{0}I^{\otimes n}+ \sum_{i=1}^{n}A_{i}Z_{i}+\sum_{i=1}^{n-1}\sum_{j=i+1}^{n}A_{ij}Z_{ij}+\sum_{i=1}^{n-2}\sum_{j=i+1}^{n-1}\sum_{l=j+1}^{n}A_{ijl}Z_{ijl}\, , 
\end{equation}
where the coefficients of \eqref{final_cub_poly} are defined by 
\begin{align*}
    \begin{split}
        A_{0} & = a_{0} + \frac{a_{1}}{2}+\frac{a_{2}}{4}\left(1 +\frac{2^{n}-1}{3(2^{n}+1)}\right)+\frac{a_{3}}{8}\left(1 +\frac{2^{n}-1}{2^{n}+1}\right) \\ 
        A_{i} & = \left(\frac{a_{1}}{2}+\frac{a_{2}}{2}+\frac{a_{3}}{8}\left(\frac{2(2^{n+1}+1)}{2^{n}+1}-\frac{2^{2n-2i+1}}{(2^{n}+1)^{2}}\right)\right)c_{i} \\ 
        A_{ij} & = \left(\frac{a_{2}}{2}+\frac{3a_{3}}{4}\right)c_{i}c_{j} \\ 
        A_{ijl} & = \frac{3a_{3}}{4}c_{i}c_{j}c_{l} \, . 
    \end{split}
\end{align*}

\end{document}